\documentclass[a4paper,UKenglish]{lipics-v2019}

\usepackage{graphicx}
\usepackage{amsmath,latexsym,amssymb,tikz,xspace,graphicx,stfloats,graphics,latexsym,amssymb}
\usetikzlibrary{calc}
\usetikzlibrary{decorations.pathreplacing}
\usetikzlibrary{decorations.markings}
\usepackage{colortbl}
\usepackage[]{color}

\title{The Fluted Fragment with Transitivity}

\author{Ian Pratt-Hartmann}
{University of Warsaw, Poland/University of Opole, Poland/University of Manchester, UK}
{ipratt@cs.man.ac.uk}{}{}

\author{Lidia Tendera
}{University of Opole, Poland}{tendera@math.uni.opole.pl}{}{}

\authorrunning{I. Pratt-Hartmann and L. Tendera}

\Copyright{Ian Pratt-Hartmann and Lidia Tendera}



\relatedversion{This is an extended version  of the MFCS 2019 paper.}

\nolinenumbers 

\newif\ifdraftpaper
\draftpaperfalse

\begin{document}
	
	\maketitle
	
	
\setlength{\abovedisplayskip}{3pt}
\setlength{\belowdisplayskip}{3pt}
\setlength{\abovedisplayshortskip}{0pt}
\setlength{\belowdisplayshortskip}{0pt}

	\begin{abstract}
		We study the satisfiability problem for the fluted fragment extended with transitive relations. We show that the logic enjoys the finite model property when only one transitive relation is available. On the other hand we show that the satisfiability problem is undecidable already for  the two-variable fragment of the logic in the presence of three transitive relations. 
	\end{abstract}

	\definecolor{GreenYellow}   {cmyk}{0.15,0,0.69,0}
	\definecolor{Yellow}        {cmyk}{0,0,1,0}
	\definecolor{myYellow}      {cmyk}{0,0,0.15,0}
	
	\definecolor{Goldenrod}     {cmyk}{0,0.10,0.84,0}
	\definecolor{Dandelion}     {cmyk}{0,0.29,0.84,0}
	\definecolor{Apricot}       {cmyk}{0,0.32,0.52,0}
	\definecolor{Peach}         {cmyk}{0,0.50,0.70,0}
	\definecolor{Melon}         {cmyk}{0,0.46,0.50,0}
	\definecolor{myMelon}       {cmyk}{0,0.1,0.1,0}
		
	\definecolor{YellowOrange}  {cmyk}{0,0.42,1,0}
	\definecolor{Orange}        {cmyk}{0,0.61,0.87,0}
	\definecolor{BurntOrange}   {cmyk}{0,0.51,1,0}
	\definecolor{Bittersweet}   {cmyk}{0,0.75,1,0.24}
	\definecolor{RedOrange}     {cmyk}{0,0.77,0.87,0}
	\definecolor{Mahogany}      {cmyk}{0,0.85,0.87,0.35}
	\definecolor{Maroon}        {cmyk}{0,0.87,0.68,0.32}
	\definecolor{BrickRed}      {cmyk}{0,0.89,0.94,0.28}
	\definecolor{Red}           {cmyk}{0,1,1,0}
	\definecolor{OrangeRed}     {cmyk}{0,1,0.50,0}
	\definecolor{RubineRed}     {cmyk}{0,1,0.13,0}
	\definecolor{WildStrawberry}{cmyk}{0,0.96,0.39,0}
	\definecolor{Salmon}        {cmyk}{0,0.53,0.38,0}
	\definecolor{CarnationPink} {cmyk}{0,0.63,0,0}
	\definecolor{Magenta}       {cmyk}{0,1,0,0}
	\definecolor{VioletRed}     {cmyk}{0,0.81,0,0}
	\definecolor{Rhodamine}     {cmyk}{0,0.82,0,0}
	\definecolor{Mulberry}      {cmyk}{0.34,0.90,0,0.02}
	\definecolor{RedViolet}     {cmyk}{0.07,0.90,0,0.34}
	\definecolor{Fuchsia}       {cmyk}{0.47,0.91,0,0.08}
	\definecolor{myFuchsia}     {cmyk}{0.17,0.18,0,0.04}

	\definecolor{Lavender}      {cmyk}{0,0.48,0,0}
	\definecolor{myLavender}      {cmyk}{0,0.1,0,0}

	\definecolor{Thistle}       {cmyk}{0.12,0.59,0,0}
	\definecolor{Orchid}        {cmyk}{0.32,0.64,0,0}
	\definecolor{DarkOrchid}    {cmyk}{0.40,0.80,0.20,0}
	\definecolor{Purple}        {cmyk}{0.45,0.86,0,0}
	\definecolor{Plum}          {cmyk}{0.50,1,0,0}
	\definecolor{Violet}        {cmyk}{0.79,0.88,0,0}
	\definecolor{RoyalPurple}   {cmyk}{0.75,0.90,0,0}
	\definecolor{BlueViolet}    {cmyk}{0.86,0.91,0,0.04}
	\definecolor{Periwinkle}    {cmyk}{0.57,0.55,0,0}
	\definecolor{CadetBlue}     {cmyk}{0.62,0.57,0.23,0}
	\definecolor{CornflowerBlue}{cmyk}{0.65,0.13,0,0}
	\definecolor{MidnightBlue}  {cmyk}{0.98,0.13,0,0.43}
	\definecolor{NavyBlue}      {cmyk}{0.94,0.54,0,0}
	\definecolor{RoyalBlue}     {cmyk}{1,0.50,0,0}
	\definecolor{Blue}          {cmyk}{1,1,0,0}
	\definecolor{Cerulean}      {cmyk}{0.94,0.11,0,0}
	\definecolor{Cyan}          {cmyk}{1,0,0,0}
	\definecolor{ProcessBlue}   {cmyk}{0.96,0,0,0}
	\definecolor{SkyBlue}       {cmyk}{0.62,0,0.12,0}
	\definecolor{Turquoise}     {cmyk}{0.85,0,0.20,0}
	\definecolor{TealBlue}      {cmyk}{0.86,0,0.34,0.02}
	\definecolor{Aquamarine}    {cmyk}{0.82,0,0.30,0}
	\definecolor{BlueGreen}     {cmyk}{0.85,0,0.33,0}
	\definecolor{Emerald}       {cmyk}{1,0,0.50,0}
	\definecolor{JungleGreen}   {cmyk}{0.99,0,0.52,0}
	\definecolor{SeaGreen}      {cmyk}{0.69,0,0.50,0}
	\definecolor{Green}         {cmyk}{1,0,1,0}
	\definecolor{ForestGreen}   {cmyk}{0.91,0,0.88,0.12}
	\definecolor{PineGreen}     {cmyk}{0.92,0,0.59,0.25}
	\definecolor{LimeGreen}     {cmyk}{0.50,0,1,0}
	\definecolor{YellowGreen}   {cmyk}{0.44,0,0.74,0}
	\definecolor{SpringGreen}   {cmyk}{0.26,0,0.76,0}
	\definecolor{OliveGreen}    {cmyk}{0.64,0,0.95,0.40}
	\definecolor{RawSienna}     {cmyk}{0,0.72,1,0.45}
	\definecolor{Sepia}         {cmyk}{0,0.83,1,0.70}
	\definecolor{Brown}         {cmyk}{0,0.81,1,0.60}
	\definecolor{Tan}           {cmyk}{0.14,0.42,0.56,0}
	\definecolor{Gray}          {cmyk}{0,0,0,0.50}
	\definecolor{Black}         {cmyk}{0,0,0,1}
	\definecolor{White}         {cmyk}{0,0,0,0}


\newcommand{\fA}{\mathfrak{A}}
\newcommand{\fB}{\mathfrak{B}}%
\newcommand{\fC}{\mathfrak{C}}%
\newcommand{\fD}{\mathfrak{D}}%
\newcommand{\fG}{\mathfrak{G}}%
\newcommand{\fI}{\mathfrak{I}}%

\newcommand{\ft}{\mathfrak{t}} 

\newcommand{\cA}{\mathcal{A}}%
\newcommand{\cB}{\mathcal{B}}%
\newcommand{\cC}{\mathcal{C}}%
\newcommand{\cD}{\mathcal{D}}%
\newcommand{\cE}{\mathcal{E}}%

\newcommand{\fod}{\ensuremath{\text{FO\/}^2}\xspace}
\newcommand{\fodt}{\ensuremath{\text{FO\/}^{2}_{1,<}}\xspace}

\renewcommand{\phi}{\varphi} 

\newcommand{\set}[1]{\ensuremath{\{ #1 \}}}  
\newcommand{\tuple}[1]{\langle#1\rangle}
\newcommand{\sizeof}[1]{|\!|#1|\!|}

\newcommand{\EQ}{\ensuremath{{\mathcal EQ}}}
\newcommand{\Sat}{\ensuremath{\textit{Sat}}}
\newcommand{\FinSat}{\ensuremath{\textit{FinSat}}}

\newcommand{\FO}{\mbox{\bf FO}}
\newcommand{\FOt}{\mbox{$\mbox{\rm FO}^2$}}

\newcommand{\FOtEC}{\mbox{$\mbox{\rm EC}^2$}}
\newcommand{\FOtECth}{\mbox{$\mbox{\rm EC}^2_3$}}
\newcommand{\FOtECt}{\mbox{$\mbox{\rm EC}^2_2$}}
\newcommand{\FOtECo}{\mbox{$\mbox{\rm EC}^2_1$}}
\newcommand{\FOtECk}{\mbox{$\mbox{\rm EC}^2_k$}}
\newcommand{\GFk}{\mbox{$\mbox{\rm GF}^k$}}
\newcommand{\GFt}{\mbox{$\mbox{\rm GF}^2$}}
\newcommand{\GFthree}{\mbox{$\mbox{\rm GF}^3$}}
\newcommand{\GF}{\mbox{$\mbox{\rm GF}$}}
\newcommand{\Ct}{{\mathcal C}^2}
\newcommand{\GCt}{{\mathcal {GC}}^2}

\newcommand{\FLt}{{\mathcal{FL}}^2}
\newcommand{\FL}{{\mathcal{ FL}}}
\newcommand{\FLotrans}{{\mathcal{FL}1\mbox{\rm T}}}
\newcommand{\FLotranso}{{\mathcal{FL}^11\mbox{\rm T}}} 
\newcommand{\FLotranst}{{\mathcal{FL}^21\mbox{\rm T}}} 
\newcommand{\FLotranstMinus}{{\mathcal{FL}^{2}1\mbox{\rm T}^{\/u}}} 
\newcommand{\FLotransm}{{\mathcal{FL}^{m}1\mbox{\rm T}}} 
\newcommand{\FLotransmpo}{{\mathcal{ FL}^{m+1}1\mbox{\rm T}}} 

\newcommand{\FLtwotrans}{{\mathcal{FL}^22\mbox{\rm T}}} 	
\newcommand{\FLthreetrans}{{\mathcal{FL}^23\mbox{\rm T}}}	
\newcommand{\FLmtrans}[1]{{\mathcal{FL}^m{#1}\mbox{\rm T}}}	
\newcommand{\FLtwomtransMinus}[1]{{\mathcal{FL}^2{#1}\mbox{\rm T}^{\/u}}}	
\newcommand{\FLtwomtrans}[1]{{\mathcal{FL}^2{#1}\mbox{\rm T}}}				
\newcommand{\FLtwotransmpo}{{\mathcal{ FL}^{m+1}2\mbox{\rm T}}} 

\newcommand{\UNF}{\mbox{$\mathcal{UNF}$}}
\newcommand{\FOk}{\mbox{$\mathcal{FO}^k$}}

\newcommand{\FOtEth}{\mbox{$\mbox{\rm EQ}^2_3$}}
\newcommand{\FOtEt}{\mbox{$\mbox{\rm EQ}^2_2$}}
\newcommand{\FOtEo}{\mbox{$\mbox{\rm EQ}^2_1$}}
\newcommand{\FOtEk}{\mbox{$\mbox{\rm EQ}^2_k$}}

\newcommand{\NLogSpace}{\textsc{NLogSpace}}
\newcommand{\NP}{\textsc{NPTime}}
\newcommand{\NPTime}{\textsc{NPTime}}
\newcommand{\PTime}{\textsc{PTime}}
\newcommand{\PSpace}{\textsc{PSpace}}
\newcommand{\ExpTime}{\textsc{ExpTime}}
\newcommand{\ExpSpace}{\textsc{ExpSpace}}
\newcommand{\NExpTime}{\textsc{NExpTime}}
\newcommand{\TwoExpTime}{2\textsc{-ExpTime}}
\newcommand{\TwoNExpTime}{2\textsc{-NExpTime}}
\newcommand{\ThreeNExpTime}{3\textsc{-NExpTime}}
\newcommand{\APSpace}{\textsc{APSpace}}
\newcommand{\Tower}{\textsc{Tower}}


\newcommand{\N}{{\mathbb N}}
\newcommand{\Q}{{\mathbb Q}}
\newcommand{\Z}{{\mathbb Z}}

\newcommand{\tp}{\ensuremath{\mbox{\rm tp}}}
\newcommand{\ftp}{\ensuremath{\mbox{\rm ftp}}}
\newcommand{\stp}{\ensuremath{\mbox{\rm stp}}} 

\newcommand{\Tp}{\ensuremath{\mbox{{\bf{tp}}}}}
\newcommand{\tpSource}{\ensuremath{\mbox{tp}_1}}
\newcommand{\tpTarg}{\ensuremath{\mbox{tp}_2}}

\newcommand{\vtw}{\ensuremath{\mathsf{v}}}
\newcommand{\htw}{\ensuremath{\mathsf{h}}}

\newcommand{\rt}{\ensuremath{\mathsf{rt}}}
\newcommand{\lt}{\ensuremath{\mathsf{lt}}}
\newcommand{\up}{\ensuremath{\mathsf{up}}}
\newcommand{\dw}{\ensuremath{\mathsf{dn}}}

\ifdraftpaper
 \newcommand{\nb}[1]{$|$\marginpar{\scriptsize\raggedright\textcolor{red}{#1}}}
 \else
\newcommand{\nb}[1]{}
\fi
 




\newcommand{\restr}{\!\!\restriction\!\!}
\newcommand{\calEo}{1{\cal EQ}}
\newcommand{\calEt}{2{\cal EQ}}
\newcommand{\calEth}{3{\cal EQ}}
\newcommand{\calEk}{k{\cal EQ}}
\newcommand{\calTo}{1{\cal TR}}
\newcommand{\calTt}{2{\cal TR}}

\newcommand{\mynegsp}{\!\!\!\!}

\newcommand{\towerOrder}[2]{\ensuremath{\mathfrak{T}(#1,#2)}}

%
%

\addtolength{\itemsep}{0.5\baselineskip}

\section{Introduction}
\label{sec:intro}
The \textit{fluted fragment}, here denoted $\FL$, is a fragment of first-order logic in which, roughly speaking, the order of quantification of variables coincides with the order in which those variables appear as arguments of predicates. The allusion is presumably architectural: we are invited to think of arguments of predicates as being `lined up' in columns. The following formulas are sentences of $\FL$
\begin{align}
& \mbox{
	\begin{minipage}{10cm}
	\begin{tabbing}
	No student admires every professor\\
	$\forall x_1 (\mbox{student}(x_1) \rightarrow \neg \forall x_2 (\mbox{prof}(x_2) \rightarrow \mbox{admires}(x_1, x_2)))$
	\end{tabbing}
	\end{minipage}
}
\label{eq:eg1}\\
& \mbox{
	\begin{minipage}{10cm}
	\begin{tabbing}
	No lecturer introduces any professor to every student\\
	$\forall x_1 ($\=$\mbox{lecturer}(x_1) \rightarrow$
	$\neg \exists x_2 ($\=$\mbox{prof}(x_2)
	\wedge \forall x_3 (\mbox{student}(x_3) \rightarrow \mbox{intro}(x_1,x_2,x_3))))$,
	\end{tabbing}
	\end{minipage}
}
\label{eq:eg2}
\end{align}
with the `lining up' of variables illustrated in Fig.~\ref{fig:lining}. By contrast, none of the formulas
%
\begin{align*}
& \forall x_1 . r(x_1, x_1)\\
& \forall x_1 \forall x_2 (r(x_1, x_2) \rightarrow r(x_2,x_1))\\
& \forall x_1 \forall x_2 \forall x_3 (r(x_1,x_2) \wedge r(x_2,x_3) \rightarrow r(x_1,x_3)),
\end{align*}
expressing, respectively, the reflexivity, symmetry and transitivity of the relation $r$, is fluted, as the atoms involved cannot be arranged so that their argument sequences `line up' in the fashion of Fig.~\ref{fig:lining}. 
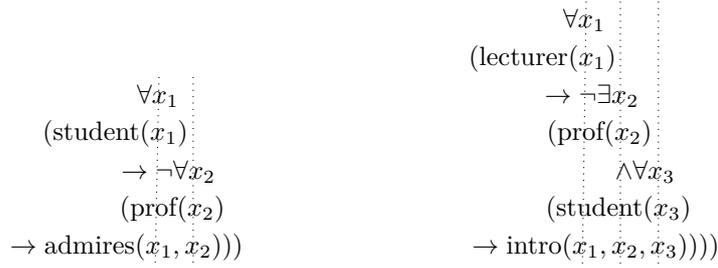
\begin{figure}
\begin{center}
\begin{tikzpicture}[scale=0.5]
\draw[dotted] (3.35,7.5) -- (3.25,2.5);
\draw[dotted] (4.25,7.5) -- (4.25,2.5);
\draw (3.325,7) node {$\forall x_1$};
\draw (2.2,6) node {$(\mbox{student}(x_1)$};
\draw (3.6,5) node {$\rightarrow \neg \forall x_2$};
\draw (3.7,4) node {$(\mbox{prof}(x_2)$};
\draw (2.5,3) node {$ \rightarrow \mbox{admires}(x_1, x_2)))$};
\end{tikzpicture}
\hspace{2.5cm}
\begin{tikzpicture}[scale=0.5]
\draw[dotted] (3.35,7.5) -- (3.25,0.5);
\draw[dotted] (4.25,7.5) -- (4.25,0.5);
\draw[dotted] (5.25,7.5) -- (5.25,0.5);
\draw (3.3,7) node {$\forall x_1$};
\draw (2.2,6) node {$(\mbox{lecturer}(x_1)$};
\draw (3.5,5) node {$ \rightarrow \neg \exists x_2$};
\draw (3.7,4) node {$(\mbox{prof}(x_2)$};
\draw (4.9,3) node {$ \wedge \forall x_3$};
\draw (4.2,2) node {$(\mbox{student}(x_3)$};
\draw (3.6,1) node {$\rightarrow \mbox{intro}(x_1,x_2,x_3))))$};
\end{tikzpicture}
\end{center}
\caption{The `lining up' of variables in the fluted formulas~\eqref{eq:eg1} and~\eqref{eq:eg2}; all quantification is executed on the right-most available column.}
\label{fig:lining}
\end{figure}

The history of this fragment is somewhat tortuous.
The basic idea of {\em fluted logic} can be traced to a paper given by W.V.~Quine to the 1968 {\em International Congress of Philosophy}~\cite{purdyTrans:quine69}, in which the author defined the {\em homogeneous $m$-adic formulas}. Quine later relaxed this fragment, in the context of a discussion of predicate-functor logic, to what he called `fluted' quantificational schemata~\cite{purdyTrans:quine76b}, 
claiming that the satisfiability problem for the relaxed fragment is decidable.
The viability of the proof strategy sketched by Quine was explicitly called into question by Noah \cite{purdyTrans:noah80}, and the subject then taken up by 
W.C.~Purdy~\cite{purdyTrans:purdy96a}, who gave his own definition of `fluted formulas', proving decidability.
It is questionable whether Purdy's reconstruction is faithful to Quine's intentions: the matter is clouded by differences in the definitions of predicate functors between between~\cite{purdyTrans:noah80}
and~\cite{purdyTrans:quine76b}, both of which Purdy cites. In fact, Quine's original definition of 
`fluted' quantificational schemata appears to coincide with a logic introduced---apparently independently---by A.~Herzig~\cite{purdyTrans:herzig90}.
Rightly of wrongly, however, the name `fluted fragment' 
has now attached itself to Purdy's definition in~\cite{purdyTrans:purdy96a}; and we shall continue to use it in that way in the present article. See Sec.~\ref{sec:prelim} for a formal definition. 

To complicate matters further, Purdy claimed in~\cite{purdyTrans:purdy02} that  
$\FL$ (i.e.~the fluted fragment, in our sense, and his) has the exponential-sized model property: if a fluted formula $\phi$ is satisfiable, then it is satisfiable over a domain of size bounded by an exponential function of the number of symbols in $\phi$. Purdy concluded that the satisfiability problem for $\FL$ is \NExpTime-complete. These latter claims are false. It was shown in~\cite{purdyTrans:P-HST16} that, although
$\FL$ has the finite model property, there
is no elementary bound on the sizes of the models required, and the satisfiability problem for $\FL$ is non-elementary.
More precisely, define $\FL^m$ to be the subfragment of
$\FL$ in which at most $m$ variables (free or bound) appear. Then the satisfiability problem 
for $\FL^m$ is $\lfloor m/2 \rfloor$-\NExpTime-hard for all $m \geq 2$ and in
$(m-2)$-\NExpTime{} for all $m \geq 3$~\cite{P-HST-FLrev}.
It follows that the satisfiability problem 
for $\FL$ is \Tower-complete, in the framework of~\cite{purdyTrans:schmitz16}.
These results fix  the exact complexity 
of satisfiability of $\FL^m$ for small values of $m$.
Indeed, the satisfiability problem for \FOt, 
the two-variable fragment of first-order logic, is known to be \NExpTime-complete~\cite{purdyTrans:gkv97}, whence the corresponding problem for 
$\FL^2$ is certainly in \NExpTime. Moreover, 
for $0 \leq m \leq 1$, $\FL^m$ coincides with the $m$-variable fragment of first-order logic,
whence its satisfiability problem is \NPTime-complete. Thus,
taking $0$-\NExpTime{} to mean \NPTime, we see that the satisfiability problem 
for $\FL^m$ is $\lfloor m/2 \rfloor$-\NExpTime-complete, at least for $m\leq 4$.

The focus of the present paper is what happens when we add to the fluted fragment the 
ability to stipulate that certain designated binary relations are \textit{transitive}, or are
{\em equivalence relations}. The motivation comes from analogous results obtained for other
decidable fragments of first-order logic.  Consider basic propositional modal logic K.
Under the standard translation into first-order logic (yielded by Kripke semantics), we can regard K
as a fragment of first-order logic---indeed as a fragment of $\FLt$.
From basic modal logic K, we obtain the logic K4 under the supposition that the accessibility relation
on possible worlds is transitive, and the logic S5 under the supposition that it is an equivalence relation: it is well-known that the satisfiability problems for K and K4 are \PSpace-complete, whereas that for
S5 is \NPTime-complete~\cite{purdyTrans:ladner77}. (For analogous results on {\em graded} modal logic, see~\cite{purdyTrans:kph09}.)
Closely related are also description logics (cf.~\cite{purdyTrans:2003handbook}) with {\em role hierarchies} and {\em transitive roles}. In particular, the description logic  $\mathcal{SH}$, which has the finite model property, is an $\ExpTime$-complete fragment of $\FL$ with transitivity. 
Similar investigations have been carried out in respect of \FOt{}, which has the finite model property and
whose satisfiability problem, as just mentioned, is \NExpTime-complete. The finite model property is lost when one transitive relation or two equivalence relations are allowed. 
For equivalence, everything is known: the (finite) satisfiability problem for \FOt{} in the presence of a single equivalence relation 
remains \NExpTime-complete, but this increases to \TwoNExpTime-complete in the presence of two equivalence relations~\cite{purdyTrans:kmp-ht,purdyTrans:KO12}, and becomes undecidable with three. For transitivity, we
have an incomplete picture:
the {\em finite} satisfiability problem for \FOt{} in the presence with a single transitive relation in decidable in \ThreeNExpTime~\cite{purdyTrans:ph18}, while the decidability of the satisfiability problem remains open (cf.~\cite{ST-FO2T}); the corresponding problems with two transitive relations 
are both undecidable~\cite{purdyTrans:KT09}. 

Adding equivalence relations to the fluted fragment poses no new problems. Existing results on
of \FOt{} with two equivalence relations can be used to show that the satisfiability and finite
satisfiability problems for $\FL$\nb{L for FL we need the resolution technique from this paper. Or?} (not just $\FLt$) with \textit{two} equivalence relations are decidable. Furthermore,
the proof that the corresponding problems for
\FOt{} in the presence of \textit{three} equivalence relations are undecidable can easily be seen to apply also to $\FLt$. 
On the other hand, the situation with transitivity is much less clear. In particular, it is not known to the present authors whether the description logic $\mathcal{SHI}$, the extension of $\mathcal{SH}$ where also role {\em inverses} can be used (a feature not expressible in $\FL$), enjoys the finite model property. 
Some indication that flutedness interacts in interesting ways with transitivity is provided by known complexity results on various extensions of
guarded two-variable fragment with transitive relations.\nb{L something could perhaps be shortened in the following text} 
The {\em guarded fragment}, denoted \GF{}, is that fragment
of first-order logic in which all quantification
is of either of the forms 
$\forall \bar{v}(\alpha \rightarrow \psi)$ or $\exists \bar{v}(\alpha \wedge \psi)$,
where $\alpha$ is an atomic formula (a so-called {\em guard}) 
featuring all free variables of $\psi$. The
{\em guarded two-variable fragment}, denoted \GFt, is the intersection of
\GF{} and \FOt. It is straightforward to show that the addition
of two transitive relations to \GFt{} yields a logic whose satisfiability 
problem is undecidable. However, as long as the distinguished 
transitive relations appear only in guards, we can extend the whole
of \GF{} with any number of transitive relations, yielding the so-called 
{\em guarded fragment with transitive guards}, whose satisfiability 
problem is in \TwoExpTime~\cite{purdyTrans:st04}. Intriguingly, in the
two-variable case, we obtain a reduction in complexity if we
require transitive relations in guards to point {\em forward}---i.e.~allowing only $\forall v(t(u,v) \rightarrow \psi)$
rather than $\forall v(t(v,u) \rightarrow \psi)$, and similarly
for existential quantification. Thus, the extension of
\GFt{} with (any number of) transitive guards has a \TwoExpTime-complete
satisfiability problem; however, the corresponding problem under 
the  restriction to one-way transitive guards is \ExpSpace-complete~\cite{purdyTrans:kieronski06}. Since the above-mentioned extensions of \GFt{} also do not enjoy the finite model property, the satisfiability and the finite satisfiability problems do not coincide. 
Decidability and complexity bounds for the finite satisfiability problems are shown in \cite{purdyTrans:KT09,purdyTrans:KT18}. 

We show in the sequel that 
$\FL$ in the presence of a single transitive relation has the finite model property.
On the other hand, the satisfiability problems for 
$\FL$ in the presence of three transitive relations are undecidable even for the intersection of $\FL$ with $\GFt$. (Indeed, the same holds even when one of these transitive relations is constrained to be the identity.)
The status of the decidability of $\FL$ with just two transitive relations remains open.


\section{Preliminaries}
\label{sec:prelim}
Unless explicitly stated to the contrary, the fragments of first-order logic considered here do not contain equality. We employ purely relational signatures, i.e.~no individual constants or function symbols. We
do, however, allow 0-ary relations (proposition letters).

Let $\bar{x}_\omega= x_1, x_2, \ldots$ be a fixed sequence of variables.
		%
		%
		%
We define the sets of formulas $\FL^{[m]}$ (for $m \geq 0$) by structural induction as follows:
(i) any atom $\alpha(x_\ell, \ldots, x_m)$, where $x_\ell, \dots, x_m$ is a contiguous subsequence of $\bar{x}_\omega$,
		is in $\FL^{[m]}$;
		(ii) $\FL^{[m]}$ is closed under boolean combinations;
		(iii) if $\phi$ is in $\FL^{[m+1]}$, then $\exists x_{m+1} \phi$ and $\forall x_{m+1} \phi$
		are in $\FL^{[m]}$.
		The set of \textit{fluted formulas} is defined as \smash{$\FL = \bigcup_{m\geq 0} \FL^{[m]}$}. A \textit{fluted sentence} is a fluted formula with no free variables. 
		Thus, when forming Boolean combinations in the fluted fragment, all the combined formulas must have as
		their free variables some suffix of some prefix $x_1, \dots, x_m$ of $\bar{x}_\omega$; and, when quantifying, only the last variable in this sequence may be bound. 
Note also that proposition letters (0-ary predicates)
may, according to the above definitions, be combined freely with formulas: if $\phi$ is in 
$\FL^{[m]}$, then so, for example, is $\phi \wedge P$, where $P$ is a proposition letter. 
	
Denote by $\FL^{m}$ the sub-fragment of $\FL$ consisting of
those formulas featuring at most $m
$ variables, free or bound.
Do not confuse $\FL^{m}$ (the set of fluted formulas with at most $m$ variables, free or bound) with $\FL^{[m]}$ (the set of fluted formulas with {\em free} variables $x_\ell, \dots, x_m$). These are, of course, quite different. For example, \eqref{eq:eg1} is in $\FL^{2}$, and \eqref{eq:eg2} is in $\FL^{3}$, but they  
are both in $\FL^{[0]}$. Note that $\FL^{m}$ cannot include predicates of arity greater than $m$. 

For $m \geq 2$, denote by $\FLmtrans{k}$ the $m$-variable fluted fragment $\FL^m$ together with $k$ distinguished transitive relations. In addition, denote by $\FLtwomtransMinus{k}$ 
the sub-fragment of 
$\FLtwomtrans{k}$ in which no binary predicates occur except the $k$ distinguished transitive ones.

\section{The decidability of fluted logic with one transitive relation}\label{sec:onetrans}

In this section, we show that the logic $\FLotrans$, the fluted fragment together with a single distinguished transitive relation $\ft$, has the finite model property.
We proceed in stages. 
First, we show that $\FLotranstMinus$ has a doubly exponential-sized model property. Next, we show that $\FLotranst$ has a triply exponential-sized model property, via an exponential-sized reduction to $\FLotranstMinus$.
Finally, for $m \geq 2$, we provide an exponential-sized reduction of the satisfiability problem for 
$\FLotransmpo$ to the corresponding problem for $\FLotransm$, showing that, if the target of the reduction has a model of size $N$, then the source has a model of size 
$O(2^N)$.
The satisfiability problems considered here will all have at least exponential complexity. Therefore,
we may assume without loss of generality in this section that all signatures feature no 0-ary predicates,
since their truth values can simply be guessed.

\subsection{The logic $\FLotranstMinus$}
Fix some signature $\Sigma$ of unary predicates. We consider $\FLotranstMinus$-formulas
over the signature $\Sigma \cup \{\ft\}$, where $\ft$ is the distinguished transitive predicate. (Thus, $\ft\not\in \Sigma$.)
By a 1-{\em type over} $\Sigma$, we mean a maximal consistent conjunction of literals $\pm p(x)$, where $p \in \Sigma$. If $\fA$ is a structure interpreting $\Sigma \cup \set{\ft}$, any element $a \in A$ satisfies a unique 1-type over $\Sigma$;
we denote it $\tp^\fA[a]$. Since $\Sigma$ will not vary, we typically omit reference to it when speaking of 1-types.
We use the letters 
$\pi$ and $\pi'$ always to range over 1-types
and $\mu$ always to range over arbitrary 
quantifier-free $\Sigma$-formulas involving just the variable $x$. We write
$\pi(y)$ to indicate the result of substituting $y$ everywhere for $x$ in 
$\pi$, and similarly for $\pi'$ and $\mu$. 

Call a $\FLotranstMinus$-formula over $\Sigma\cup \{\ft\}$ {\em basic} if it is of one of the forms
$$\exists x . \mu \hspace{0.8cm}  \forall x . \mu \hspace{0.8cm}  \forall x (\pi \rightarrow \exists y (\mu(y) \wedge \pm \ft(x,y))) \hspace{0.8cm}
\forall x (\pi \rightarrow \forall y (\pi'(y) \rightarrow \pm \ft(x,y))).$$
The following Lemma is a version of the familiar `Scott normal form' for $\FLt$ from \cite{purdyTrans:scott62}.
%
\begin{lemma}
	Let $\phi$ be a $\FLotranstMinus$-sentence. There exists a set $\Psi$ of
	basic $\FLotranstMinus$-formulas with the following properties:
	\textup{(}i\textup{)} $\models \left(\bigwedge \Psi \right) \rightarrow
	\phi$; \textup{(}ii\textup{)} if $\phi$ has a model, then so has $\Psi$;
	\textup{(}iii\textup{)} 
	$\sizeof{\Psi}$ is bounded by a polynomial function of $\sizeof{\phi}$.
	\label{lma:nf}
\end{lemma}

We say that a {\em super-type} {\em over} $\Sigma$ is a pair $\langle \pi, \Pi \rangle$, where $\pi$ is a 1-type over $\Sigma$ and $\Pi$ a set of 1-types over $\Sigma$. If $\fA$ is a structure interpreting the signature $\Sigma \cup \set{\ft}$ and $a \in A$, the super-type of $a$ in $\fA$, denoted $\stp^\fA[a]$, is the pair 
$\langle \tp^\fA[a], \Pi \rangle$, where $\Pi = \set{\tp^\fA[b] \mid \fA \models \ft[a,b]}$. 
Intuitively, a super-type is a description of an element in a structure specifying that
element's 1-type together with the 1-types of those elements to which it is related by $\ft$. 
If $S$ is a set of super-types, we write $\tp(S) = \set{\pi \mid 
	\langle \pi, \Pi \rangle \in S \mbox{ for some } \Pi}$. 
Since $\Sigma$ will not vary, we again omit reference to it when speaking of super-types.
By a {\em certificate}, we mean a pair $C = (S, \ll)$, where $S$ is a set 
of super-types and $\ll$ is a transitive relation on $\tp(S)$ satisfying the following conditions:
\begin{description}
	\item [(C1)] if $\langle \pi, \Pi \rangle \in S$ and $\pi' \in \Pi$, then there exists $\langle \pi', \Pi' \rangle \in S$ with $\Pi' \subseteq \Pi$; 
	\item [(C2)] if $\pi \ll \pi'$, $\langle \pi, \Pi \rangle \in S$ and
	$\langle \pi', \Pi' \rangle \in S$, then $\set{\pi'} \cup \Pi' \subseteq \Pi$.  
\end{description}
If $\fA$ is a structure, then the {\em certificate of} $\fA$, denoted
$C(\fA)$, is the pair $(S, \ll)$, where $S = \set{\stp^\fA[a] \mid a \in A}$
is the set of super-types realized in $\fA$, and $\pi \ll \pi'$ if and only
if $\pi$ and $\pi'$ are realized in $\fA$ and $\fA \models \forall x(\pi \rightarrow \forall y(\pi'(y) \rightarrow \ft(x,y)))$.
Intuitively, a certificate is a description of a structure listing the realized super-types
and containing a partial order which specifies when all elements realizing one 1-type
are related by $\ft$ to all elements realizing another 1-type.

\begin{lemma}
	If $\fA$ is any structure interpreting $\Sigma\cup \{\ft\}$, $C(\fA)$ is a certificate.
	\label{lma:isaCerificate}
\end{lemma}
\begin{proof}
	Write $C(\fA) = (S,\ll)$. Obviously $\ll$ is  transitive.
	We must check ({\bf C1}) and ({\bf C2}).
	
	\noindent
	({\bf C1}): Suppose $\langle \pi, \Pi \rangle \in S$ and $\pi' \in \Pi$. Let $a$ be
	such that $\stp^\fA[a] = \langle \pi, \Pi \rangle$. Then there exists a $b \in A$
	such that $\tp^\fA[b] = \pi'$ and $\fA \models \ft[a,b]$. Let 
	$\stp^\fA[b] = \langle \pi', \Pi' \rangle$.
	
	\noindent
	({\bf C2}): 
	Suppose $\langle \pi, \Pi \rangle \in S$ and 
	$\langle \pi', \Pi' \rangle \in S$ with $\pi \ll \pi'$, and let $a, b \in A$ be such that 
	$\stp^\fA[a] = \langle \pi, \Pi \rangle$ and $\stp^\fA[b] = \langle \pi', \Pi' \rangle$.
	Since $\pi \ll \pi'$, by construction of $C(\fA)$, we have 
	$\fA \models \forall x (\pi \rightarrow \forall y(\pi'(y) \rightarrow \ft(x,y)))$,
	whence it is immediate that $\pi' \in \Pi$ and $\Pi' \subseteq \Pi$.
\end{proof}
If $C = (S, \ll)$ is a certificate, and $\psi$ a basic $\FLotranstMinus$-formula,
we define the relation $C \models \psi$ to hold provided the following six conditions are satisfied.
The motivation for this definition is provided by Lemmas~\ref{lma:correct} and~\ref{lma:complete}.
\begin{enumerate}[(i)]
	\item if $\psi$ is of the form $\forall x (\pi \rightarrow \exists y(\mu(y) \wedge \ft(x,y)))$, then, for all $\Pi$ such that $\langle \pi, \Pi \rangle \in S$,
	there exists $\pi' \in \Pi$ such that 
	$\models \pi' \rightarrow \mu$;
	\item  if $\psi$ is of the form $\forall x (\pi \rightarrow \forall y(\pi'(y) \rightarrow \ft(x,y)))$ and $\pi, \pi' \in \tp(S)$, then
	$\pi \ll \pi'$;
	\item if $\psi$ is of the form $\forall x (\pi \rightarrow \exists y(\mu(y) \wedge \neg \ft(x,y)))$, then, for all $\langle \pi, \Pi \rangle \in S$,
	there exists $\langle \pi', \Pi' \rangle \in S$ such that
	$\models \pi' \rightarrow \mu$ and there exists 
	no $\alpha \in \set{\pi} \cup \Pi$ such that $\alpha \ll  \pi'$;
	\item  if $\psi$ is of the form $\forall x (\pi \rightarrow \forall y(\pi'(y) \rightarrow \neg \ft(x,y)))$, then, for all 
	$\langle \pi, \Pi \rangle \in S$, $\pi' \not \in \Pi$;
	\item if $\psi$ is of the form $\exists x . \mu$, then there exists
	$\langle \pi, \Pi \rangle \in S$ such that $\models \pi \rightarrow \mu$;
	\item if $\psi$ is of the form $\forall x . \mu$, then, for all
	$\langle \pi, \Pi \rangle \in S$, $\models \pi \rightarrow \mu$.
\end{enumerate}
\begin{lemma}
	Let $\fA$ be a structure interpreting $\Sigma\cup \{\ft\}$ and let
	$\psi$ be a basic $\FLotranstMinus$-formula over $\Sigma \cup \{\ft\}$.
	If $\fA \models \psi$, then $C(\fA) \models \psi$. 
	\label{lma:correct}
\end{lemma}
\begin{proof}
	We write $C(\fA) = (S,\ll)$ and consider the possible forms of $\psi$ in turn.
	\begin{itemize}
		\item 
		$\psi = \forall x (\pi \rightarrow \exists y(\mu(y) \wedge \ft(x,y)))$: 
		Suppose $\langle \pi, \Pi \rangle \in S$. Then there exists $a \in A$ with
		$\stp^\fA[a] =  \langle \pi, \Pi \rangle$. Since $\fA \models \psi$, 
		choose $b \in A$ such that $\fA \models \mu[b]$ and $\fA \models \ft[a,b]$,
		and let $\tp^\fA[b] = \pi'$. Then $\models \pi' \rightarrow \mu$ and $\pi' \in \Pi$, as required. 
		\item 
			$\psi = \forall x (\pi \rightarrow \forall y(\pi'(y) \rightarrow \ft(x,y)))$: 
		It is immediate by the construction of $C(\fA)$ that, if $\pi, \pi' \in \tp(S)$, then $\pi \ll \pi'$;
			\item 
	$\psi = \forall x (\pi \rightarrow \exists y(\mu(y) \wedge \neg \ft(x,y)))$: 
	Suppose $\langle \pi, \Pi \rangle \in S$. Then there exists $a \in A$ with
	$\stp^\fA[a] =  \langle \pi, \Pi \rangle$. Since $\fA \models \psi$, 
	choose $b \in A$ such that $\fA \models \mu[b]$ and $\fA \not 
	\models \ft[a,b]$,
	and let $\tp^\fA[b] = \pi'$, so that $\models \pi' \rightarrow \mu$.
	We require only to show that there exists no $\alpha \in \set{\pi} \cup \Pi$ such that $\alpha \ll \pi'$. Suppose, for contradiction, that such an $\alpha$ exists. By ({\bf C1}), $\alpha \in \tp(S)$. 
	If $\alpha = \pi$, then, by the definition of $\ll$, we have  
	$\fA \models \forall x (\pi \rightarrow \forall y (\pi'(y) \rightarrow \ft(x,y)))$, which contradicts the supposition that $\fA \not 
	\models \ft[a,b]$.
	If $\alpha \in \Pi$, then, by the definition of 
	$\Pi$ and $\ll$, we have an element $a' \in A$ such that
	$\tp^\fA[a'] = \alpha$, $\fA \models \ft[a,a']$ and 
	$\fA \models \forall x (\alpha \rightarrow \forall y (\pi'(y) \rightarrow \ft(x,y)))$, which again contradicts the supposition that $\fA \not 
	\models \ft[a,b]$.
	\item 
	$\psi = \forall x (\pi \rightarrow \forall y(\pi'(y) \rightarrow \neg \ft(x,y)))$: 
	Suppose $\langle \pi, \Pi \rangle \in S$ and let $a \in A$ be such that
	$\stp^\fA[a] = \langle \pi, \Pi \rangle$. Since $\fA \models \psi$,
	we have $\pi' \not \in \Pi$. 
		\item $\psi = \exists x. \mu$ or $\psi = \forall x. \mu$. Immediate by construction of $S$.
		\end{itemize}
\vspace{-15pt}\end{proof}

\begin{lemma}
	If $C = (S,\ll)$ is a certificate, then there exists a structure $\fA$ over 
	a domain of cardinality $2|S|$ such that, for any basic 
	$\FLotranstMinus$-formula $\psi$ over $\Sigma$, $C \models \psi$ implies
	$\fA \models \psi$. 
	\label{lma:complete}
\end{lemma}
\begin{proof}
	Define 
	$	A^+ =  \set{a^+_{\pi, \Pi} \mid \langle \pi, \Pi \rangle \in S}$ and  
	$A^- =  \set{a^-_{\pi, \Pi} \mid \langle \pi, \Pi \rangle \in S}$, 
	where the various $a^+_{\pi, \Pi}$ and $a^-_{\pi, \Pi}$ are some objects
	(assumed distinct), and set $A = A^+ \cup A^-$. Define the binary relations 
$T_1 =  \set{\langle a^\pm_{\pi, \Pi}, a^+_{\pi', \Pi'}\rangle \mid \set{\pi'} \cup \Pi' \subseteq \Pi}$ and 
$T_2 =  \set{\langle a^\pm_{\pi, \Pi}, a^\pm_{\pi', \Pi'}\rangle \mid \pi \ll \pi'}$,
%
	and let $T$ be the transitive closure of $T_1 \cup T_2$. 
	Intuitively, we may think of the elements $a^+_{\pi', \Pi'}$ as witnessing existential formulas of the form
	$\exists y (\mu(y) \wedge \ft(x,y))$, where $\models \pi' \rightarrow \mu$, and
	of the elements $a^-_{\pi', \Pi'}$ as witnessing existential formulas of the form
	$\exists y (\mu(y) \wedge \neg \ft(x,y))$. Now define
	$\fA$ on the domain $A$ by setting $\tp^\fA[a^\pm_{\pi,\Pi}] = \pi$ for
	all $\langle \pi, \Pi \rangle \in S$, and by setting $\ft^\fA = T$.
	
	We observe that if $a = a^\pm_{\pi, \Pi}$ and $b = a^\pm_{\pi', \Pi'}$ 
	are elements of $A$ such that $a$ is related to $b$ by either $T_1$ or $T_2$,
	then $\set{\pi'} \cup \Pi' \subseteq \Pi$. Indeed, for $T_1$, this
	is immediate by definition; and for $T_2$, it follows from property ({\bf C2}) 
	of certificates. It follows by induction that, if
	$a$ is related to $b$ by $T$, then $\set{\pi'} \cup \Pi' \subseteq \Pi$. 
	To prove the lemma, we consider the possible forms of $\psi$ in turn.
	\begin{itemize}
\item 	$\psi = \forall x (\pi \rightarrow \exists y(\mu(y) \wedge \ft(x,y)))$: 
	Suppose  $a = a^\pm_{\pi, \Pi}$. Since $C \models \psi$, there exists 
	$\pi' \in \Pi$ such that $\pi' \rightarrow \mu$. By ({\bf C1}), there exists
	$\langle \pi', \Pi' \rangle \in S$ such that $\Pi' \subseteq \Pi$. 
	Letting $b = a^+_{\pi', \Pi'}$, we have
	that $a$ is related to $b$ by $T_1$. But then $\fA \models \ft[a,b]$ 
	and $\fA \models \mu[b]$ by construction of $\fA$.
	
\item 	$\psi = \forall x (\pi \rightarrow \forall y(\pi'(y) \rightarrow \ft(x,y)))$: 
	Since $C \models \psi$, we have $\pi \ll \pi'$. Suppose now
	$a = a^\pm_{\pi, \Pi}$ and $b = a^\pm_{\pi', \Pi'}$. 
	Thus, $a$ is related to $b$ by $T_2$, and so
	by construction of $\fA$, $\fA \models \ft[a,b]$. 
	
\item 	$\psi = \forall x (\pi \rightarrow \exists y(\mu(y) \wedge \neg \ft(x,y)))$: 
	Suppose  $a = a^\pm_{\pi, \Pi}$. Since $C \models \psi$, there exists 
	$\langle \pi', \Pi' \rangle \in S$ such that $\pi' \rightarrow \mu$, and such that there is no $\alpha \in \set{\pi} \cup \Pi$ with $\alpha \ll \pi'$. Now let $b = a^-_{\pi',\Pi'}$. By construction of $\fA$,
	$\fA \models \mu[b]$. It suffices to show that $\fA \not \models \ft[a,b]$.
	For otherwise, by the definition of $T$, there exists  a chain of
	elements $a = a_1, \dots, a_{m} = b$ with each related to the next by either 
	$T_1$ or $T_2$ and with $a_{m-1}$ related to $a_m$ by $T_2$. 
	(Notice that nothing can be related by $T_1$ to $b = a^-_{\pi',\Pi'}$.)
	Writing $a_{m-1}= a^\pm_{\alpha, \Pi''} \in S$, we see that 
	$\alpha \ll \pi'$, and, moreover, that 
	$a$ is either identical to $a_{m-1}$, or related to it by $T$.
	As we observed above, if $a^\pm_{\pi, \Pi}$ is related to $a^\pm_{\alpha, \Pi''}$
	by $T$, then $\alpha \in \Pi$. Thus, either way, $\alpha \in \set{\pi} \cup \Pi$. 
	But we are supposing that no such $\alpha$ exists.
	
\item 	$\psi = \forall x (\pi \rightarrow \forall y(\pi'(y) \rightarrow \neg \ft(x,y)))$: 
	Suppose $a = a^\pm_{\pi, \Pi}$ and $b = a^\pm_{\pi', \Pi'}$ are elements of $A$. We observed above that, if $a$ is related to $b$ by $T$, then $\pi' \in \Pi$, contradicting the assumption that $C \models \psi$.
	Thus, 
	by construction of $\fA$, $\fA \not \models \ft[a,b]$. 
	
\item
	$\psi = \exists x. \mu$ or $\psi = \forall x. \mu$. Immediate by construction of $\fA$.
\end{itemize}
\vspace{-15pt}
\end{proof}
Since the number of super-types over $\Sigma$ is bounded by $2^{(2^{|\Sigma|}+|\Sigma|)}$, Lemmas~\ref{lma:nf}--\ref{lma:complete} yield:
\begin{lemma}
	If $\phi$ is a satisfiable formula of $\FLotranstMinus$, then $\phi$ has a model
	of size at most doubly exponential in $\sizeof{\phi}$. Hence
	the satisfiability problem for $\FLotranstMinus$ is in $\TwoNExpTime$.
	\label{lma:FL1transoUpper}
\end{lemma}
\begin{proof}

A structure $\fA$ can be guessed and verified to be a model of any $m$-variable first-order formula $\phi$ in time $O(|\phi|\cdot|A|^m)$~\cite{purdyTrans:vardi95}.
\end{proof}

\subsection{The logics $\FLotransm$ for $m \geq 2$}
Let $\Sigma$ be a signature of predicates of positive arity, excluding $\ft$. 
An atomic formula of $\FLotransm$ involving a predicate from $\Sigma \cup \set{\ft}$
will be called a \textit{fluted $m$-atom} \textit{over} $\Sigma \cup \set{\ft}$. 
A {\em fluted $m$-literal} is a fluted $m$-atom or the negation thereof. A {\em fluted $m$-clause} is a disjunction of fluted $m$-literals. We allow the absurd formula $\bot$ (i.e.~the empty disjunction) to count as a fluted $m$-clause. Thus, any literal of a fluted $m$-clause has arguments $x_h, \dots, x_m$, in that order, for some $h$ ($1 \leq h \leq m)$.
When writing fluted $m$-clauses, we silently remove bracketing, re-order literals and delete duplicated literals as necessary. 
A \textit{fluted $m$-type} is a maximal consistent set of fluted $m$-literals;
where convenient, we identify fluted $m$-types with their conjunctions. 
If $\fA$ is a structure interpreting $\Sigma \cup \set{\ft}$, any tuple $a_1, \dots,
a_m$ from $A$ satisfies a unique fluted $m$-type; we denote it $\ftp^\fA[a_1, \dots, a_m]$.
Note that a fluted
1-type  over $\Sigma \cup \set{\ft}$ coincides with what we earlier called a 1-type over $\Sigma$.
Reference to the signature $\Sigma \cup \set{\ft}$
will as usual be suppressed when clear from context. Predicates in $\Sigma$ will be referred to as {\em non-distinguished}. Our strategy will be to reduce the (finite) satisfiability problem for 
$\FLotransm$ to that for $\FLotranst$
(Lemma~\ref{lma:eliminate}), and thence to
that for $\FLotranstMinus$ (Lemma~\ref{lma:eliminateNonTrans}), which 
we have already dealt with (Lemma~\ref{lma:FL1transoUpper}).
\nb{II: slightly elaborated}

A $\FLotransm$-formula $\phi$ ($m \geq 2$) is in \textit{clause normal form} 
if it is of the form
\begin{equation}
\forall x_1 \cdots x_{m} . \Omega \ \wedge
{\bigwedge_{i=1}^{s} }\forall x_1 \cdots x_{m-1} \left(\alpha_i \rightarrow \exists x_m . \Gamma_i\right) \wedge
{\bigwedge_{j=1}^{t}} \forall x_1 \cdots x_{m-1} (\beta_j \rightarrow \forall x_m . \Delta_j),
\label{eq:cnf}
\end{equation}
where $\Omega, \Gamma_1, \dots, \Gamma_s, \Delta_1, \dots, \Delta_t$ are sets of fluted $m$-clauses,
and $\alpha_1, \dots, \alpha_s$, $\beta_1, \dots, \beta_t$ fluted \mbox{$(m-1)$}-atoms. We refer to $\forall x_1 \cdots x_{m} . \Omega$ as the \textit{static conjunct}
of $\phi$, to conjuncts of the form $\forall x_1 \cdots x_{m-1} \left(\alpha_i \rightarrow \exists x_m \Gamma_i\right)$ as the \textit{existential conjuncts}
of $\phi$, and to conjuncts of the form $\forall x_1 \cdots x_{m-1} (\beta_j \rightarrow \forall x_m . \Delta_j)$ as the \textit{universal conjuncts} of $\phi$. 

Using the same techniques as for Lemma~\ref{lma:nf}, we can transform any $\FLotransm$-formula into clause normal form.
\begin{lemma}
	Let $\phi$ be an $\FLotransm$-formula, $m \geq 2$. There exists an
	$\FLotransm$-formula $\psi$ in clause normal form such that:
	\textup{(}i\textup{)} $\models \psi \rightarrow \phi$; and \textup{(}ii\textup{)} if $\phi$ has a model then so has $\psi$;
	\textup{(}iii\textup{)} 
	$\sizeof{\psi}$ is bounded by a polynomial function of $\sizeof{\phi}$.
	\label{lma:cnf}
\end{lemma}

For fragments of first-order logic not involving equality, we are free to duplicate any element 
$a$ in a structure $\fA$. More formally, we have the following lemma,
which will be used as a step in the ensuing argument.
\begin{lemma}
	Let $\fA$ be any structure, and let $z >0$. There exists a structure $\fB$ such that \textup{(}i\textup{)} if $\phi$ is any first-order formula without equality, then	$\fA \models \phi$ if and only if $\fB \models \phi$;  \textup{(}ii\textup{)} 
	$|B| = z \cdot |A|$; and \textup{(}iii\textup{)}
	if $\psi(x_1, \dots, x_{m-1}) = \exists x_m . \chi(x_1, \dots, x_m)$ is a first-order formula without equality, and $\fB \models \psi[b_1, \dots, b_{m-1}]$, then there exist at least $z$ distinct elements $b$ of $B$ such that $\fB \models \chi[b_1, \dots, b_{m-1}, b]$. 
	\label{lma:multiply}
\end{lemma}

Keeping the signature
$\Sigma$ fixed, we employ the standard apparatus of resolution theorem-proving to eliminate
non-distinguished predicates of arity 2 or more.
Suppose $p \in \Sigma$ is a 
predicate of arity
$m$, and let $\gamma'$ and $\delta'$ be fluted $m$-clauses
over $\Sigma$. Then, $\gamma = p(x_1, \dots, x_m) \vee \gamma'$ and $\delta = \neg p(x_1, \dots, x_m) \vee \delta'$ are also fluted $m$-clauses, as indeed is $\gamma' \vee \delta'$. In that case, we call $\gamma' \vee \delta'$ a {\em fluted resolvent} of $\gamma$ and $\delta$, and we say that $\gamma' \vee \delta'$ is {\em obtained by fluted resolution} from $\gamma$ and $\delta$ \textit{on} $p(x_1, \dots, x_m)$. Thus, fluted resolution 
is simply a restriction of the familiar resolution rule from first-order logic to the case where the resolved-on literals have maximal arity, $m$,
and (in the case $m =2$) do not feature the distinguished predicate $\ft$. It may be helpful to note the following at this point: 
(i) if $\gamma$ and $\delta$ resolve to form $\epsilon$, then $\models \forall x_1 \cdots \forall x_m (\gamma \wedge \delta \rightarrow \epsilon)$;
(ii) the fluted resolvent of two fluted $m$-clauses may or may not
involve predicates of arity $m$;
(iii) in fluted resolution, the arguments of the literals in the fluted $m$-clauses undergo no change when forming the resolvent;
(iv) if the fluted $m$-clause $\gamma$ involves no predicates
of arity $m$, then it cannot undergo fluted resolution at all.

If $\Gamma$ is a set of fluted $m$-clauses, denote by $\Gamma^*$ the smallest set of fluted $m$-clauses including $\Gamma$ and closed under fluted resolution. If $\Gamma = \Gamma^*$, we say that it is \textit{closed under fluted resolution}. We further denote by $\Gamma^\circ$ the result of deleting
from $\Gamma^*$ any clause involving a \textit{non-distinguished} predicate of arity $m$. Observe that, since
all fluted $m$-atoms involving predicates of {\em non}-maximal arity are of the form $p(x_h, \dots, x_m)$ for some $h \geq 2$,
it follows that
$\Gamma^\circ$ features the variable $x_1$ {\em only} in the case $m=2$, and even then 
only in literals of the form $\pm \ft(x_1,x_2)$.

The following lemma is, in effect, nothing more than the familiar completeness theorem for (ordered) propositional resolution. 
 Due to space limits the proof is given in Section~\ref{app:resolution}.
\begin{lemma}
	Let $\Gamma$ be a set of fluted $m$-clauses over a signature $\Sigma \cup \set{\ft}$, let $\Sigma'$ be the result of removing all predicates
	of maximal arity $m$  from $\Sigma$, and let $\tau^-$ be a fluted $m$-type over $\Sigma' \cup \set{\ft}$.  
	If $\tau^-$ is consistent with $\Gamma^\circ$, then there exists a fluted $m$-type $\tau$ over the signature $\Sigma\cup \set{\ft}$ 
	such that $\tau \supseteq \tau^-$ and $\tau$ is consistent with $\Gamma$.
	\label{lma:resolution}
\end{lemma}
The following lemma employs a technique from~\cite{purdyTrans:P-H18}
to eliminate binary predicates.
\begin{lemma}
	Let $\phi$ be an $\FLotranst$-formula in clause normal form
	over a signature $\Sigma\cup \{\ft\}$, and suppose that $\phi$ has $s$ existential
	and $t$ universal conjuncts. Then there exists a clause normal form $\FLotranstMinus$-formula $\phi'$ over a signature $\Sigma'\cup \{\ft\}$ such that:
	\textup{(}i\textup{)} $\phi'$ has at most $2^{t} s$ existential and $2^{t}$ universal conjuncts;
	\textup{(}ii\textup{)} $|\Sigma'| \leq |\Sigma| + 2^t(s+1)$;
	\textup{(}iii\textup{)} if $\phi$ has a model, so does $\phi'$; and
	\textup{(}iv\textup{)} if $\phi'$ has a model of size $M$, then $\phi$ has a model of size at most $sM$.
	\label{lma:eliminateNonTrans}
\end{lemma}


\begin{proof}
Let $\phi = \forall x_1 x_{2} . \Omega \ \wedge
{\bigwedge_{i=1}^{s} }\forall x_1 \left(p_i(x_1) \rightarrow \exists x_2 . \Gamma_i\right) \wedge
{\bigwedge_{j=1}^{t}} \forall x_1 (q_j(x_1) \rightarrow \forall x_2 . \Delta_j)$,
%
where $\Omega, \Gamma_1, \dots, \Gamma_s, \Delta_1, \dots, \Delta_t$ are sets of fluted 2-clauses,
and $p_1, \dots, p_s$, $q_1, \dots, q_t$ unary predicates. Write $T = \{1, \dots, t\}$.
For all $i$ ($1 \leq i \leq s$) and all $J \subseteq T$, let $p_{i,J}$ and $q_J$ be new unary
predicates.
The intended interpretation of $p_{i,J}(x_1)$ is ``$x_1$ satisfies $p_i$, and also satisfies $q_j$ for every $j \in J$;'' and
the intended interpretation of $q_{J}(x_1)$ is ``$x_1$ satisfies $q_j$ for every $j \in J$.''
Let $\phi'$ be the conjunction of the sentences:\\ 
\hspace*{1em}(a) $	\bigwedge_{i=1}^{s} \bigwedge_{J \subseteq T} \forall x_2
((p_i(x_2) \wedge \bigwedge_{j \in J} q_j(x_2))  \rightarrow p_{i,J}(x_2))$, \\ 
\hspace*{1em}(b) $\bigwedge_{J \subseteq T}\forall x_2
((\bigwedge_{j \in J} q_j(x_2)) \rightarrow q_{J}(x_2))$,\\ 
\hspace*{1em}(c) $\bigwedge_{i=1}^{s} \bigwedge_{J \subseteq T} \forall x_1
\left( p_{i,J}(x_1)\rightarrow \exists x_2 \left(\Gamma_i \cup \Omega \cup \bigcup \{\Delta_j \mid j \in J \}\right)^\circ \right) $, and\\ 
\hspace*{1em}(d)
$\bigwedge_{J \subseteq T}\forall x_1 
\left(q_{J}(x_1) \rightarrow \forall x_2 \left(\Omega \cup \bigcup \{\Delta_j \mid j \in J \}\right)^\circ \right).$\\
%
%
%
%
Observe that $\phi'$ contains no non-distinguished binary predicates, and hence is in $\FLotranstMinus$.
Clearly, $\phi'$ satisfies properties (i) and (ii). To show (iii), suppose $\fA \models \phi$, and 
let $\fA'$ be the structure obtained by interpreting the predicates $p_{i,J}$ and $q_J$ as suggested above. To show (iv), suppose $\phi'$ has a model of size $M$. By Lemma~\ref{lma:multiply},
$\phi'$ has a model $\fB$ of size $sM$ in which witnesses for all the conjuncts in (c) are duplicated $s$ times. We need to show that $\fB$ can be expanded to a model of $\phi$. Fix $a \in B$ and suppose
$a$ satisfies $p_1$. Let $J$ be the set of indices $j$ such that $a$ satisfies $q_j$. By (a), putting $i = 1$,
$a$ satisfies $p_{1,J}$, whence, by (c), there exists $b$ such that the pair $\langle a,b\rangle$
satisfies $\left( \Gamma_1 \cup \Omega \cup \bigcup \{\Delta_j \mid j \in J \}\right)^\circ$. But 
Lemma~\ref{lma:resolution} guarantees that we can expand $\fB$ by interpreting the non-distinguished 
binary predicates so that $\langle a,b\rangle$  satisfies 
$\Gamma_1 \cup \Omega \cup \bigcup \{\Delta_j \mid j \in J \}$. Because of the duplication of witnesses, we 
can repeat with $p_2, \dots, p_s$, choosing a fresh witness each time, so as to avoid clashes.
Do this for all elements $a$. At the end of the process, the partially defined expansion of $\fB$ satisfies all the existential conjuncts of $\phi$,
and violates none of the universal or static conjuncts. A precisely similar argument shows that we may complete the expansion so that no universal or static conjuncts of $\phi$ are violated. 
See Section~\ref{app:eliminateNonTrans} in the Appendix for details.
\end{proof}

Thus, at the expense of an exponentially larger signature, we have reduced the
(finite) satisfiability problem for $\FLotranst$ to that for $\FLotranstMinus$. 
By Lemmas~\ref{lma:FL1transoUpper} and~\ref{lma:eliminateNonTrans}, we obtain
\begin{lemma}
	Let $\phi$ be a $\FLotranst$-formula. If $\phi$ is satisfiable, then $\phi$ has a model of size at most triply exponential in $\sizeof{\phi}$. Hence the satisfiability problem for $\FLotranst$ is in \ThreeNExpTime.
	\label{lma:FLotranstUpper}
\end{lemma}

We now establish the finite model property for the whole of $\FLotrans$ by eliminating variables
from $\FLotransmpo$, where $m \geq 2$, one at a time. The proof of the following Lemma is similar to the proof of Lemma~\ref{lma:eliminateNonTrans} and has been relegated to the Appendix.
\begin{lemma}
	Let $\phi$ be a clause normal form $\FLotransmpo$-formula \textup{(}$m \geq 2$\textup{)} 
	over a signature $\Sigma\cup \{\ft\}$, and suppose that $\phi$ has $s$ existential conjuncts
	and $t$ universal conjuncts. 
	Then there exists a clause normal form $\FLotransm$-formula $\phi'$ over a signature $\Sigma'\cup \{\ft\}$ such that the following hold:
	\textup{(}i\textup{)} $\phi'$ has at most $2^{t} s$ existential and $2^{t}$ universal conjuncts;
	\textup{(}ii\textup{)} $|\Sigma'| \leq |\Sigma| + 2^t(s+1)$;
	\textup{(}iii\textup{)} if $\phi$ has a model, so does $\phi'$; and
	\textup{(}iv\textup{)} if $\phi'$ has a model of size $M$, then $\phi$ has a model of size at most $sM$.
	\label{lma:eliminate}
\end{lemma}
\begin{theorem}
	Let $\phi$ be a $\FLotransm$-formula for $m \geq 2$. If $\phi$ is satisfiable, then $\phi$
	 has a model
	 of size at most $(m+1)$-tuply exponential in $\sizeof{\phi}$. Hence the satisfiability problem for $\FLotransm$ is in non-deterministic $(m+1)$-tuply exponential time.
	\label{theo:FLotransmUpper}
\end{theorem}
\begin{proof}
	Induction on $m$. The case $m=2$ is Lemma~\ref{lma:FLotranstUpper}. The inductive step is Lemma~\ref{lma:eliminate}.	
\end{proof}	
We mentioned in Sec.~\ref{sec:intro} that~\cite{purdyTrans:P-HST16} establishes a lower bound of $\lfloor m/2 \rfloor$-\NExpTime-hard for the satisfiability
problem for $\FL^{m}$. For $m \geq 3$, this appears to be the best available lower bound on the corresponding problem for $\FLotransm$.
Thus, a gap remains between the best available upper and lower complexity bounds. \nb{L I think Emanuel's proof could be written in FL$^4$, so it doesn't give new bounds.}
Certainly, it follows that the satisfiability problem for $\FLotrans$ is \Tower-complete, as for $\FL$.

%

\section{Fluted Logic with more Transitive Relations}\label{sec:undecidable}
In this section we show two undecidability results for the fluted fragment with two variables, $\FLt$, extended with more transitive relations, that have been informally announced in \cite{Tendera18}. We employ the apparatus of tiling systems. 

A \emph{tiling system} is a tuple $\boldsymbol{\mathcal{C}}=({\mathcal C}, {\mathcal C}_H, {\mathcal C}_V)$, where
$\mathcal C$ is a finite set of {\em tiles},
and
${\mathcal C}_H$, ${\mathcal C}_V \subseteq {\mathcal C} \times {\mathcal C}$ are the  {\em horizontal} and {\em vertical} constraints.

Let $S$ be either of the spaces $\N \times \N$, $\Z\times \Z$  or $\Z_t \times \Z_t$. 
A tiling system $\boldsymbol{\mathcal{C}}$ \emph{tiles} $S$, if there exists a function $\rho: S \rightarrow {\cal C}$ such that for all $(p,q)\in S$: 
$	(\rho(p,q), \rho(p+1,q)) \in {\cal C}_H$ and $(\rho(p,q), \rho(p,q+1)) \in {\cal C}_V$. 
The following problems are known to be undecidable (cf.~e.g.~\cite{purdyTrans:BGG97}):
\begin{itemize}
	\item
Given a tiling system $\boldsymbol{\mathcal{C}}$ determine if $\boldsymbol{\mathcal{C}}$ tiles $\Z\times\Z$, or $\N\times \N$.
	\item
	Given a tiling system $\boldsymbol{\mathcal{C}}$ determine if $\boldsymbol{\mathcal{C}}$ tiles $\Z_t\times\Z_t$, for some $t\geq 1$.
\end{itemize}

In this section we first prove the following theorem.
\begin{theorem}
	The satisfiability problem for $\FLthreetrans$, the two-variable fluted fragment with three transitive relations, is undecidable. 
	\label{th:three}
\end{theorem}

\begin{proof}
	Suppose the signature contains transitive relations $b$ (black), $g$ (green)  and $r$ (red), and additional unary predicates $e$, $e'$, $f$, $l$, $c_{i,j}$ ($0\leq i\leq 5$, $0\leq j\leq 2$) and $d_{i,j}$ ($0\leq i \leq 2$, $0\leq j \leq 5$); we refer to the  $c_{i,j}$'s and to the $d_{i,j}$'s as {\em colours}. 
	
	We reduce from the $\N\times\N$ tiling problem. We first write a formula $\phi_{grid}$ that captures several properties of the intended expansion of the $\N\times\N$ grid  as shown 
	in Fig.~\ref{fig:grid-three-transitive:Left}. There the predicates $c_{i,j}$  and $d_{i,j}$ together define a partition of the universe as follows: 
	an element $(k,k')$ with $k'>k$ (i.e.~in the yellow region, above the diagonal) satisfies $c_{i,j}$ with $i=k\mod 6$, $j=k'\mod 3$, and an  element $(k,k')$ with $k\geq k'$ (i.e.~in the pink region, on or below the diagonal) satisfies $d_{i,j}$ with $i=k\mod 3, j=k'\mod 6$. 
	Paths of the same transitive relation have length at most 7 and follow one of four designated patterns. Remaining unary predicates mark the following elements: $l$---left column, $f$---bottom row, $e$---main diagonal, and  $e'$---elements with coordinates $(k,k+1)$. 
	\begin{figure}[thb]
		\begin{center}
			\begin{subfigure}[t]{0.5\textwidth}
				\begin{center}
					\resizebox{!}{4cm}
					{
						\begin{tikzpicture}[scale=0.95]
						
						\clip (-0.8,-0.8) rectangle (11.3,7.3);

						\path [fill=myYellow] (0,0.5) to (8,8.5) to (0,8.5) to (0,0.5);
						\path [fill=myMelon] (0,0.5) to  (8,8.5) to (14.4,8.3) to (14.4,0) to (0,0) to (0,0.5);
						
						\draw [help lines] (0,0) grid (18,18);
						
						\foreach \x in {0,1,2,3,4,5,6,7,8,9,10,11,12,13,14,15,16,17}
						\foreach \y in {0,1,2,3,4,5,6,7,8,9,10,11,12,13,14,15,16,17}
						{
							\filldraw[fill=white] (\x, \y) circle (0.1); 
						}

						\foreach \x in {1,7,13} \foreach \y in {\x,\x+3,\x+6,\x+9,\x+12,\x+15} \foreach \s in {0.1}
						{	\filldraw[fill=red] (\x, \y) circle (0.1);
							\draw [->, rounded corners, ultra thick, red]  (\x+\s,\y) -- (\x+\s,\y-1+\s) -- (\x+1-\s,\y-1+\s) -- (\x+2-\s,\y-1+\s) -- (\x+2-\s,\y-\s) -- (\x+1-\s,\y-\s) -- (\x+1-\s,\y+1-\s) -- (\x+\s,\y+1-\s);
							\draw [->, ultra thick, dashed, red] (\x+\s,\y+\s) -- (\x+\s,\y+1-\s-\s);
						}
						
						\foreach \x in {3,9,15} \foreach \y in {\x,\x+3,\x+6,\x+9,\x+12,\x+15} \foreach \s in {0.1}
						{	\filldraw[fill=black] (\x, \y) circle (0.1); 
							\draw [->, rounded corners, ultra thick, black]  (\x+\s,\y) -- (\x+\s,\y-1+\s) -- (\x+1-\s,\y-1+\s) -- (\x+2-\s,\y-1+\s) -- (\x+2-\s,\y-\s) -- (\x+1-\s,\y-\s) -- (\x+1-\s,\y+1-\s) -- (\x+\s,\y+1-\s);
							\draw [->, ultra thick, dashed, black] (\x+\s,\y+\s) -- (\x+\s,\y+1-\s-\s);
						}
						\foreach \x in {5,11} \foreach \y in {\x,\x+3,\x+6,\x+9,\x+12,\x+15} \foreach \s in {0.1}
						{	\filldraw[fill=green] (\x, \y) circle (0.1);
							\draw [->, rounded corners, ultra thick, green]  (\x+\s,\y) -- (\x+\s,\y-1+\s) -- (\x+1-\s,\y-1+\s) -- (\x+2-\s,\y-1+\s) -- (\x+2-\s,\y-\s) -- (\x+1-\s,\y-\s) -- (\x+1-\s,\y+1-\s) -- (\x+\s,\y+1-\s);
							\draw [->, ultra thick, dashed, green] (\x+\s,\y+\s) -- (\x+\s,\y+1-\s-\s);
						}


						\foreach \y in {2,8,14} \foreach \x in {\y+2,\y+5,\y+8,\y+11,\y+14} \foreach \s in {0.1}
						{	\filldraw[fill=red] (\x, \y) circle (0.1);
							\draw [->, rounded corners, ultra thick, red]  (\x+\s,\y) -- (\x+\s,\y-1+\s) -- (\x+\s,\y-2+\s) -- (\x+1-\s,\y-2+\s) -- (\x+2-\s,\y-2+\s) -- (\x+2-\s,\y-1-\s) -- (\x+1-\s,\y-1-\s) -- (\x+1-\s,\y-\s);
							\draw [->, ultra thick, dashed, red] (\x+\s,\y-\s) -- (\x+1-\s-\s,\y-\s);
						}	
						
						\foreach \y in {4,10,16} \foreach \x in {\y+2,\y+5,\y+8,\y+11,\y+14} \foreach \s in {0.1}
						{	\filldraw[fill=black] (\x, \y) circle (0.1);
							\draw [->, rounded corners, ultra thick, black]  (\x+\s,\y) -- (\x+\s,\y-1+\s) -- (\x+\s,\y-2+\s) -- (\x+1-\s,\y-2+\s) -- (\x+2-\s,\y-2+\s) -- (\x+2-\s,\y-1-\s) -- (\x+1-\s,\y-1-\s) -- (\x+1-\s,\y-\s);
							\draw [->, ultra thick, dashed, black] (\x+\s,\y-\s) -- (\x+1-\s-\s,\y-\s);
						}	
						
						\foreach \y in {6,12,18} \foreach \x in {\y+2,\y+5,\y+8,\y+11,\y+14} \foreach \s in {0.1}
						{	\filldraw[fill=green] (\x, \y) circle (0.1);
							\draw [->, rounded corners, ultra thick, green]  (\x+\s,\y) -- (\x+\s,\y-1+\s) -- (\x+\s,\y-2+\s) -- (\x+1-\s,\y-2+\s) -- (\x+2-\s,\y-2+\s) -- (\x+2-\s,\y-1-\s) -- (\x+1-\s,\y-1-\s) -- (\x+1-\s,\y-\s);
							\draw [->, ultra thick, dashed, green] (\x+\s,\y-\s) -- (\x+1-\s-\s,\y-\s);
						}	
						
						
						
						\foreach \y in {3,9,15} \foreach \x in {\y+2,\y+5,\y+8,\y+11,\y+14} \foreach \s in {0.1}
						{	\filldraw[fill=red] (\x, \y) circle (0.1);
							\draw [->, rounded corners, ultra thick,  red]  (\x,\y+\s) -- (\x-1+\s,\y+\s) -- (\x-1+\s,\y+1-\s) -- (\x-1+\s,\y+2-\s) -- (\x-\s,\y+2-\s) -- (\x-\s,\y+1-\s) -- (\x+1-\s,\y+1-\s) -- (\x+1-\s,\y+\s);
							\draw [->, ultra thick, dashed, red] (\x+\s,\y+\s) -- (\x+1-\s-\s,\y+\s);
						}
						
						\foreach \y in {5,11,17} \foreach \x in {\y+2,\y+5,\y+8,\y+11,\y+14} \foreach \s in {0.1}
						{	\filldraw[fill=black] (\x, \y) circle (0.1);
							\draw [->, rounded corners, ultra thick,  black]  (\x,\y+\s) -- (\x-1+\s,\y+\s) -- (\x-1+\s,\y+1-\s) -- (\x-1+\s,\y+2-\s) -- (\x-\s,\y+2-\s) -- (\x-\s,\y+1-\s) -- (\x+1-\s,\y+1-\s) -- (\x+1-\s,\y+\s);
							\draw [->, ultra thick, dashed, black] (\x+\s,\y+\s) -- (\x+1-\s-\s,\y+\s);
						}
						\foreach \y in {1,7,13,20} \foreach \x in {\y+2,\y+5,\y+8,\y+11,\y+14,\y+17} \foreach \s in {0.1}
						{	\filldraw[fill=green] (\x, \y) circle (0.1);
							\draw [->, rounded corners, ultra thick,  green]  (\x,\y+\s) -- (\x-1+\s,\y+\s) -- (\x-1+\s,\y+1-\s) -- (\x-1+\s,\y+2-\s) -- (\x-\s,\y+2-\s) -- (\x-\s,\y+1-\s) -- (\x+1-\s,\y+1-\s) -- (\x+1-\s,\y+\s);
							\draw [->, ultra thick, dashed, green] (\x+\s,\y+\s) -- (\x+1-\s-\s,\y+\s);
						}
						
						%
						
						\foreach \x in {2,8,14} \foreach \y in {\x,\x+3,\x+6,\x+9,\x+12} \foreach \s in {0.1}
						{	\filldraw[fill=black] (\x, \y) circle (0.1);
							\draw [->, rounded corners, ultra thick,  black]  (\x,\y+\s) -- (\x-1+\s,\y+\s) -- (\x-2+\s,\y+\s) -- (\x-2+\s,\y+1-\s) -- (\x-2+\s,\y+2-\s) -- (\x-1-\s,\y+2-\s) -- (\x-1-\s,\y+1-\s) -- (\x-\s,\y+1-\s);
							\draw [->, ultra thick, dashed, black] (\x-\s,\y+\s) -- (\x-\s,\y+1-\s-\s);
						}
						\foreach \x in {4,10,16} \foreach \y in {\x,\x+3,\x+6,\x+9,\x+12} \foreach \s in {0.1}
						{	\filldraw[fill=green] (\x, \y) circle (0.1);
							\draw [->, rounded corners, ultra thick,  green]  (\x,\y+\s) -- (\x-1+\s,\y+\s) -- (\x-2+\s,\y+\s) -- (\x-2+\s,\y+1-\s) -- (\x-2+\s,\y+2-\s) -- (\x-1-\s,\y+2-\s) -- (\x-1-\s,\y+1-\s) -- (\x-\s,\y+1-\s);
							\draw [->, ultra thick, dashed, green] (\x-\s,\y+\s) -- (\x-\s,\y+1-\s-\s);
						}
						\foreach \x in {6,12} \foreach \y in {\x,\x+3,\x+6,\x+9,\x+12} \foreach \s in {0.1}
						{	\filldraw[fill=red] (\x, \y) circle (0.1);
							\draw [->, rounded corners, ultra thick,  red]  (\x,\y+\s) -- (\x-1+\s,\y+\s) -- (\x-2+\s,\y+\s) -- (\x-2+\s,\y+1-\s) -- (\x-2+\s,\y+2-\s) -- (\x-1-\s,\y+2-\s) -- (\x-1-\s,\y+1-\s) -- (\x-\s,\y+1-\s);
							\draw [->, ultra thick, dashed, red] (\x-\s,\y+\s) -- (\x-\s,\y+1-\s-\s);
						}
						
						%
						
						
						\foreach \x in {0,3,6,9,12,15,18} \foreach \y in {0} \foreach \s in {0.1}
						{\filldraw[fill=black] (\x, \y) circle (0.1);
							\draw [->, rounded corners, ultra thick,  black]  (\x+\s,\y) -- (\x+\s,\y+1-\s) -- (\x+1-\s,\y+1-\s) -- (\x+1-\s,\y+\s);
							\draw [->, ultra thick, dashed, black] (\x+\s,\y+\s) -- (\x+1-\s-\s,\y+\s);
						}
						
						
						\foreach \y in {1,4,7,10,13,16,19} \foreach \x in {0} \foreach \s in {0.1}
						{	\filldraw[fill=green] (\x, \y) circle (0.1);
							\draw [->, rounded corners, ultra thick,  green]  (\x,\y+\s) -- (\x+1-\s,\y+\s) -- (\x+1-\s,\y+1-\s) -- (\x+\s,\y+1-\s);
							\draw [->, ultra thick, dashed, green] (\x+\s,\y+\s) -- (\x+\s,\y+1-\s-\s);
						}

						\foreach \x in {0,3,6,9,12,15}
						\foreach \y in {0} 
						{
							\filldraw[fill=black] (\x, \y) circle (0.1); 
						}

						%
						\foreach \x in {0,1,2,3,4,5,6,7,8,9,10,11,12,13,14,15,16,17} \coordinate [label=center:\x] (A) at (\x, -0.5);
						\foreach \y in {0,1,2,3,4,5,6,7,8,9,10,11,12,13,14,15,16,17}  \coordinate [label=center:\y] (A) at (-0.5,\y);
						\end{tikzpicture}
					}
				\end{center} 
				\caption{three transitive relations:  $b$,  {\color{green} $g$} and  {\color{red} $r$}. Filled nodes depict the beginning of a transitive path of the same colour; dotted lines connect the first element with the last element on such  path. }
			\label{fig:grid-three-transitive:Left}
			\end{subfigure}	
			\hspace{0.6cm} \vspace{-0.5cm}
			\begin{subfigure}[t]{0.4\textwidth}					
				\begin{center}
					\resizebox{5.5cm}{4cm}
					{
							\begin{tikzpicture}[xscale=0.8,yscale=0.8]
						\clip (-0.8,-0.8) rectangle (6.3,4.3);
						\foreach \x in {1,3,5,7}
						\foreach \y in {1,3,5,7} \foreach \z in {0.1} 
						{
							\draw[color=red,thick,-, >=latex] (\x+\z, \y+\z) -- (\x+1-\z, \y+\z) --
							(\x+1-\z, \y+1-\z) -- (\x+\z, \y+1-\z) -- (\x+\z, \y+\z);
							\draw[color=red,thick,-, >=latex] (\x+\z, \y+\z) --
							(\x+1-\z, \y+1-\z) -- (\x+\z, \y+1-\z) -- (\x+1-\z, \y+\z);
							
						}	
						
						\foreach \x in {0,2,4,6}
						\foreach \y in {0,2,4,6} \foreach \z in {0.1} 
						{
							\draw[color=blue,thick,-, >=latex] (\x+\z, \y+\z) -- (\x+1-\z, \y+\z) --
							(\x+1-\z, \y+1-\z) -- (\x+\z, \y+1-\z) -- (\x+\z, \y+\z);
							\draw[color=blue,thick,-, >=latex] (\x+\z, \y+\z) --
							(\x+1-\z, \y+1-\z) -- (\x+\z, \y+1-\z) -- (\x+1-\z, \y+\z);
						}	
						
						\foreach \x in {1,5}
						\foreach \y in {1,5} \foreach \z in {0.15} 
						{
							\draw[color=blue,thick,->, >=latex] (\x, \y) -- (\x+1-\z, \y);
							\draw[color=blue,thick,->, >=latex] (\x, \y) -- (\x, \y+1-\z);
							\draw[color=blue,thick,->, >=latex] (\x-1, \y) -- (\x-1, \y+1-\z);
							\draw[color=blue,thick,->, >=latex] (\x, \y-1) -- (\x+1-\z, \y-1);						
						}	
						
						\foreach \x in {3,7}
						\foreach \y in {3,7} \foreach \z in {0.15} 
						{
							\draw[color=blue,thick,->, >=latex] (\x, \y) -- (\x+1-\z, \y);
							\draw[color=blue,thick,->, >=latex] (\x, \y) -- (\x, \y+1-\z);
							\draw[color=blue,thick,->, >=latex] (\x-1, \y) -- (\x-1, \y+1-\z);
							\draw[color=blue,thick,->, >=latex] (\x, \y-1) -- (\x+1-\z, \y-1);						
						}

						\foreach \x in {1,5}
						\foreach \y in {3,7} \foreach \z in {0.15} 
						{
							\draw[color=blue,thick,->, >=latex] (\x+1-\z, \y) -- (\x, \y) ;
							\draw[color=blue,thick,->, >=latex] (\x, \y+1-\z) -- (\x, \y);
							\draw[color=blue,thick,->, >=latex] (\x-1, \y+1-\z) -- (\x-1, \y);
							\draw[color=blue,thick,->, >=latex] (\x+1-\z, \y-1) -- (\x, \y-1);						
						}	
						
						\foreach \x in {3,7}
						\foreach \y in {1,5} \foreach \z in {0.15} 
						{
							\draw[color=blue,thick,->, >=latex] (\x+1-\z, \y) -- (\x+\z, \y) ;
							\draw[color=blue,thick,->, >=latex] (\x, \y+1-\z) -- (\x, \y+\z);
							\draw[color=blue,thick,->, >=latex] (\x-1, \y+1-\z) -- (\x-1, \y+\z);
							\draw[color=blue,thick,->, >=latex] (\x+1-\z, \y-1) -- (\x+\z, \y-1);						
						}	
						
						%
						%

						\foreach \x in {0,1,2,3,4,5,6,7}
						\foreach \y in {0,1,2,3,4,5,6,7}
						{
							\filldraw[fill=white] (\x, \y) circle (0.1); 
							\pgfmathtruncatemacro{\abc}{mod(\y,4)} 
							\pgfmathtruncatemacro{\bcd}{mod(\x,4)} 
						}

						
						\draw[ ->, dotted, very thick] (7.5,0) -- (8,0);
						\draw[ ->, dotted, very thick] (0,7.5) -- (0,8);
						\draw[ ->, dotted, very thick] (7.5,7.5) -- (8,8);
						
						
						\foreach \x in {0,1,2,3,4,5,6,7,8,9,10,11,12,13,14,15,16,17} \coordinate [label=center:\x] (A) at (\x, -0.5);
						\foreach \y in {0,1,2,3,4,5,6,7,8,9,10,11,12,13,14,15,16,17}  \coordinate [label=center:\y] (A) at (-0.5,\y);
						
						\end{tikzpicture}

					}
				\end{center}
				\caption{two transitive relations: {\color{blue}$b$}   and  {\color{red}$r$}. Edges without arrows depict connections in both direction. 
				}
			\label{fig:grid-three-transitive:Right}
			\end{subfigure}
			
		\end{center}
		
		\caption{Expansions of the $\N\times\N$ grid in the proofs of Theorem~\ref{th:three} (a) and Theorem~\ref{th:two} (b).  
		}\label{fig:grid-three-transitive}
	\end{figure}
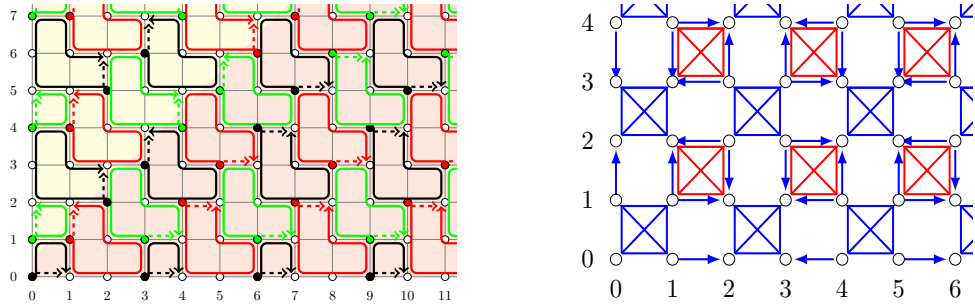

	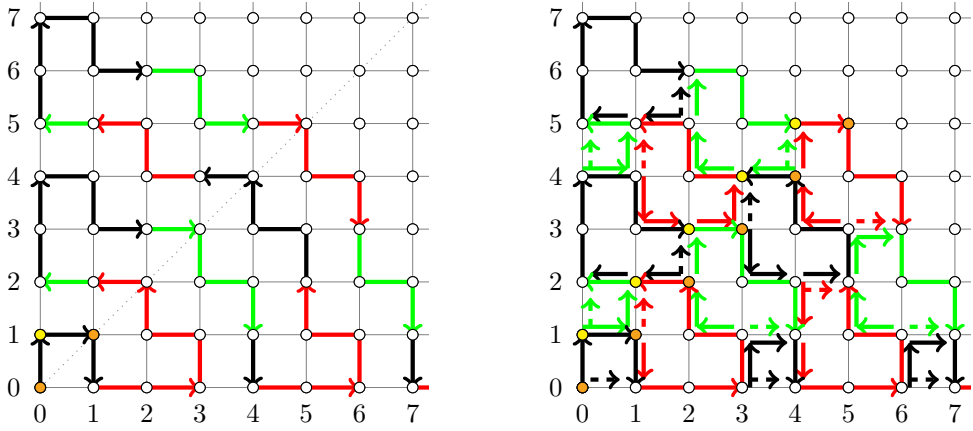
\begin{figure}[hbt]
	\begin{subfigure}[t]{0.45\textwidth}	
		\begin{center}
			\begin{minipage}{6.5cm}
				\begin{center}
					{

						\begin{tikzpicture}[scale=0.7]
						
						\clip (-0.8,-0.8) rectangle (7.3,7.3);
						
						\draw [help lines] (0,0) grid (18,18);
						\draw [gray, dotted] (0,0) -- (18,18);
						
						\foreach \x in {0,1,2,3,4,5,6,7,8,9,10,11,12,13,14,15,16,17}
						\foreach \y in {0,1,2,3,4,5,6,7,8,9,10,11,12,13,14,15,16,17}
						{
							\filldraw[fill=white] (\x, \y) circle (0.1); 
						}
						
						\draw[ultra thick, ->, black] (0,0) -- (0,1);
						
						\foreach \x in {1} \draw[ultra thick, <-, black] (\x,\x) -- (\x-1,\x);
						\foreach \x in {2,8,14} \draw[ultra thick, ->, red] (\x,\x) -- (\x-1,\x);
						\foreach \x in {3,9,15} \draw[ultra thick, <-, green] (\x,\x) -- (\x-1,\x);
						\foreach \x in {4,10,16} \draw[ultra thick, ->, black] (\x,\x) -- (\x-1,\x);
						\foreach \x in {5,11,17} \draw[ultra thick, <-, red] (\x,\x) -- (\x-1,\x);
						
						
						\foreach \y in {5,11,17} \foreach \x in {\y,\y+3,\y+6,\y+9,\y+12,\y+15} \draw[ultra thick, ->, red] (\x,\y) -- (\x,\y-1) -- (\x+1,\y-1) -- (\x+1,\y-2);
						\foreach \y in {2,8,14} \foreach \x in {\y,\y+3,\y+6,\y+9,\y+12,\y+15} \draw[ultra thick, <-, red] (\x,\y) -- (\x,\y-1) -- (\x+1,\y-1) -- (\x+1,\y-2);
						
						\foreach \y in {4,10,16} \foreach \x in {\y} 
						\draw[ultra thick, <-, black] (\x,\y) -- (\x,\y-1) -- (\x+1,\y-1) -- (\x+1,\y-2);
						
						\foreach \y in {1} \foreach \x in {\y,\y+3,\y+6,\y+9,\y+12,\y+15} \draw[ultra thick, ->, black] (\x,\y) -- (\x,\y-1);
						
						\foreach \y in {3,9,15} \foreach \x in {\y,\y+3,\y+6,\y+9,\y+12,\y+15} \draw[ultra thick, ->, green] (\x,\y) -- (\x,\y-1) -- (\x+1,\y-1) -- (\x+1,\y-2);
						
						
						\foreach \x in {3,9,15} \foreach \y in {\x+1}
						\draw[ultra thick, ->, red] (\x,\y) -- (\x-1,\y) -- (\x-1,\y+1) -- (\x-2,\y+1);
						
						\foreach \x in {4,10,16} \foreach \y in {\x+1,\x+4,\x+7,\x+10,\x+13,\x+16} \draw[ultra thick, <-, green] (\x,\y) -- (\x-1,\y) -- (\x-1,\y+1) -- (\x-2,\y+1);
						\foreach \x in {7,13} \foreach \y in {\x+1,\x+4,\x+7,\x+10,\x+13,\x+16} \draw[ultra thick, ->, green] (\x,\y) -- (\x-1,\y) -- (\x-1,\y+1) -- (\x-2,\y+1);
						\foreach \x in {1} \foreach \y in {\x+1,\x+4,\x+7,\x+10,\x+13,\x+16} \draw[ultra thick, ->, green] (\x,\y) -- (\x-1,\y);
						
						\foreach \x in {2,8,14} \foreach \y in {\x+1,\x+4,\x+7,\x+10,\x+13,\x+16} \draw[ultra thick, <-, black] (\x,\y) -- (\x-1,\y) -- (\x-1,\y+1) -- (\x-2,\y+1);
						
						\foreach \x in {0} \foreach \y in {2,5,8,11,14,17} \draw[ultra thick, ->, black] (\x,\y) -- (\x,\y+2);
						
						\foreach \y in {0} \foreach \x in {1,4,7,10,13,16} \draw[ultra thick, ->, red] (\x,\y) -- (\x+2,\y);
						
						\foreach \x in {0,1,2,3,4,5,6,7,8,9,10,11,12,13,14,15,16,17} \coordinate [label=center:\x] (A) at (\x, -0.5);
						\foreach \y in {0,1,2,3,4,5,6,7,8,9,10,11,12,13,14,15,16,17}  \coordinate [label=center:\y] (A) at (-0.5,\y);
						
						\foreach \x in {0,1,2,3,4,5,6,7,8,9,10,11,12,13,14,15,16,17}
						\foreach \y in {0,1,2,3,4,5,6,7,8,9,10,11,12,13,14,15,16,17}
						{
							\filldraw[fill=white] (\x, \y) circle (0.1); 
						}

						\foreach \x in {0,1} 
						{
							\filldraw[fill=YellowOrange] (\x, \x) circle (0.1); 
						}
						\foreach \x in {0} 
						{
							\filldraw[fill=yellow] (\x, \x+1) circle (0.1); 
						}

						\end{tikzpicture}
					}
				\end{center}
			\end{minipage}
		\end{center}
	\caption{Path starting with an initial element and generated by witnesses for conjuncts of group (3). 
	}
		\label{fig:embeddingLeft}
	\end{subfigure}
	\hspace{0.6cm}
	\begin{subfigure}[t]{0.45\textwidth}	
		\begin{center}
			\begin{minipage}{6.5cm}
				\begin{center}
					{
						\begin{tikzpicture}[scale=0.7]
						
						\clip (-0.8,-0.8) rectangle (7.3,7.3);
						
						\draw [help lines] (0,0) grid (18,18);
						
						\foreach \x in {0,1,2,3,4,5,6,7,8,9,10,11,12,13,14,15,16,17}
						\foreach \y in {0,1,2,3,4,5,6,7,8,9,10,11,12,13,14,15,16,17}
						{
							\filldraw[fill=white] (\x, \y) circle (0.1); 
						}

						
						\foreach \s in {0.15}
						{			\draw [->, ultra thick, dashed, red] (1+\s,1+\s) -- (1+\s,1+1-\s-\s);
							\draw [->, ultra thick, red] (1+\s,1-\s) -- (1+\s,0+\s);
						}
						
						\foreach \x in {1} \foreach \y in {\x+3} \foreach \s in {0.15}
						{	
							\draw [->,  ultra thick, red]  (\x+\s,\y) -- (\x+\s,\y-1+\s);
							\draw [->,  ultra thick, red] (\x+\s,\y-1+\s) -- (\x+1-\s,\y-1+\s);
							\draw [->,  ultra thick, red] (\x+1+\s,\y-1+\s) -- (\x+2-\s,\y-1+\s);
							\draw [->,  ultra thick, red] (\x+2-\s,\y-1+\s) -- (\x+2-\s,\y-\s);
							\draw [->, ultra thick, dashed, red] (\x+\s,\y+\s) -- (\x+\s,\y+1-\s-\s);
						}
						
						\foreach \x in {3} \foreach \y in {\x} \foreach \s in {0.15}
						{	
							\draw [->,  ultra thick, black]  (\x+\s,\y) -- (\x+\s,\y-1+\s);
							\draw [->,  ultra thick, black] (\x+\s,\y-1+\s) -- (\x+1-\s,\y-1+\s);
							\draw [->,  ultra thick, black] (\x+1+\s,\y-1+\s) -- (\x+2-\s,\y-1+\s);
							\draw [->, ultra thick, dashed, black] (\x+\s,\y+\s) -- (\x+\s,\y+1-\s-\s);
							
						}
						%
						

						
						\foreach \y in {2} \foreach \x in {\y+2} \foreach \s in {0.15}
						{	
							\draw [->, ultra thick, red]  (\x+\s,\y) -- (\x+\s,\y-1+\s);
							\draw [->, ultra thick, red] (\x+\s,\y-1-\s) -- (\x+\s,\y-2+\s);
							\draw [->, ultra thick, dashed, red] (\x+\s,\y-\s) -- (\x+1-\s-\s,\y-\s);
						}	

						\foreach \y in {3,9,15} \foreach \x in {\y+2} \foreach \s in {0.15}
						{	
							\draw [->, ultra thick,  red]  (\x-\s,\y+\s) -- (\x-1+\s,\y+\s);
							\draw [->, ultra thick,  red] (\x-1+\s,\y+\s) -- (\x-1+\s,\y+1-\s);
							\draw [->, ultra thick,  red] (\x-1+\s,\y+1+\s) -- (\x-1+\s,\y+2-\s);
							\draw [->, ultra thick, dashed, red] (\x+\s,\y+\s) -- (\x+1-\s-\s,\y+\s);
						}
						
						\foreach \y in {1} \foreach \x in {\y+2} \foreach \s in {0.15}
						{
							\draw [->, ultra thick,  green]  (\x-\s,\y+\s) -- (\x-1+\s,\y+\s);
							\draw [->, ultra thick,  green] (\x-1+\s,\y+\s) -- (\x-1+\s,\y+1-\s);
							\draw [->, ultra thick,  green] (\x-1+\s,\y+1+\s) -- (\x-1+\s,\y+2-\s);
							\draw [->, ultra thick, dashed, green] (\x+\s,\y+\s) -- (\x+1-\s-\s,\y+\s);
							
						}
						
						\foreach \y in {1,7,13,20} \foreach \x in {\y+5,\y+8,\y+11,\y+14,\y+17} \foreach \s in {0.15}
						{
							\draw [->, ultra thick,  green]  (\x-\s,\y+\s) -- (\x-1+\s,\y+\s);
							\draw [->, ultra thick,  green] (\x-1+\s,\y+\s) -- (\x-1+\s,\y+1-\s);
							\draw [->, ultra thick,  green] (\x-1+\s,\y+1+\s) -- (\x-1+\s,\y+2-\s);
							\draw [->, ultra thick,  green] (\x-1+\s,\y+2-\s)  -- (\x-\s,\y+2-\s);
							\draw [->, ultra thick, dashed, green] (\x+\s,\y+\s) -- (\x+1-\s-\s,\y+\s);
							
						}

						%
						
						\foreach \x in {2} \foreach \y in {\x,\x+3,\x+6,\x+9,\x+12} \foreach \s in {0.15}
						{	
							\draw [->, ultra thick,  black]  (\x-\s,\y+\s) -- (\x-1+\s,\y+\s);
							\draw [->, ultra thick,  black]  (\x-1-\s,\y+\s) -- (\x-2+\s,\y+\s);
							\draw [->, ultra thick, dashed, black] (\x-\s,\y+\s) -- (\x-\s,\y+1-\s-\s);
						}
						
						\foreach \x in {4,10,16} \foreach \y in {\x} \foreach \s in {0.15}
						{	
							\draw [->, ultra thick,  green]  (\x,\y+\s) -- (\x-1+\s,\y+\s); 
							\draw [->, ultra thick,  green] (\x-1-\s,\y+\s) -- (\x-2+\s,\y+\s);
							\draw [->, ultra thick,  green] (\x-2+\s,\y+\s) -- (\x-2+\s,\y+1-\s); 
							\draw [->, ultra thick,  green] (\x-2+\s,\y+1+\s) -- (\x-2+\s,\y+2-\s);
							\draw [->, ultra thick, dashed, green] (\x-\s,\y+\s) -- (\x-\s,\y+1-\s-\s);
						}
						%
						%
						
						
						\foreach \x in {0} \foreach \y in {0} \foreach \s in {0.15}
						{
							\draw [->, ultra thick, dashed, black] (\x+\s,\y+\s) -- (\x+1-\s-\s,\y+\s);
						}

						\foreach \x in {3,6,9,12,15,18} \foreach \y in {0} \foreach \s in {0.15}
						{
							\draw [->, ultra thick,  black]  (\x+\s,\y) -- (\x+\s,\y+1-\s);
							\draw [->, ultra thick,  black]  (\x+\s,\y+1-\s) -- (\x+1-\s,\y+1-\s);
							\draw [->, ultra thick, dashed, black] (\x+\s,\y+\s) -- (\x+1-\s-\s,\y+\s);
						}
						
						
						\foreach \y in {1,4} \foreach \x in {0} \foreach \s in {0.15}
						{	
							\draw [->, ultra thick,  green]  (\x,\y+\s) -- (\x+1-\s,\y+\s);
							\draw [->, ultra thick,  green] (\x+1-\s,\y+\s) -- (\x+1-\s,\y+1-\s);
							\draw [->, ultra thick, dashed, green] (\x+\s,\y+\s) -- (\x+\s,\y+1-\s-\s);
						}

						\foreach \x in {0,3,6,9,12,15}
						\foreach \y in {0} 
						{
							\filldraw[fill=black] (\x, \y) circle (0.1); 
						}

						\draw[ultra thick, ->, black] (0,0) -- (0,1);
						
						\foreach \x in {1} \draw[ultra thick, <-, black] (\x,\x) -- (\x-1,\x);
						\foreach \x in {2,8,14} \draw[ultra thick, ->, red] (\x,\x) -- (\x-1,\x);
						\foreach \x in {3,9,15} \draw[ultra thick, <-, green] (\x,\x) -- (\x-1,\x);
						\foreach \x in {4,10,16} \draw[ultra thick, ->, black] (\x,\x) -- (\x-1,\x);
						\foreach \x in {5,11,17} \draw[ultra thick, <-, red] (\x,\x) -- (\x-1,\x);
						
						
						\foreach \y in {5,11,17} \foreach \x in {\y,\y+3,\y+6,\y+9,\y+12,\y+15} \draw[ultra thick, ->, red] (\x,\y) -- (\x,\y-1) -- (\x+1,\y-1) -- (\x+1,\y-2);
						\foreach \y in {2,8,14} \foreach \x in {\y,\y+3,\y+6,\y+9,\y+12,\y+15} \draw[ultra thick, <-, red] (\x,\y) -- (\x,\y-1) -- (\x+1,\y-1) -- (\x+1,\y-2);
						
						\foreach \y in {4} \foreach \x in {\y} \draw[ultra thick, <-, black] (\x,\y) -- (\x,\y-1) -- (\x+1,\y-1) -- (\x+1,\y-2);
						
						\foreach \y in {1} \foreach \x in {\y,\y+3,\y+6,\y+9,\y+12,\y+15} \draw[ultra thick, ->, black] (\x,\y) -- (\x,\y-1);
						
						\foreach \y in {3,9,15} \foreach \x in {\y,\y+3,\y+6,\y+9,\y+12,\y+15} \draw[ultra thick, ->, green] (\x,\y) -- (\x,\y-1) -- (\x+1,\y-1) -- (\x+1,\y-2);
						
						
						\foreach \x in {3,9,15} \foreach \y in {\x+1} \draw[ultra thick, ->, red] (\x,\y) -- (\x-1,\y) -- (\x-1,\y+1) -- (\x-2,\y+1);
						
						
						\foreach \x in {4,10,16} \foreach \y in {\x+1,\x+4,\x+7,\x+10,\x+13,\x+16} \draw[ultra thick, <-, green] (\x,\y) -- (\x-1,\y) -- (\x-1,\y+1) -- (\x-2,\y+1);
						\foreach \x in {7,13} \foreach \y in {\x+1,\x+4,\x+7,\x+10,\x+13,\x+16} \draw[ultra thick, ->, green] (\x,\y) -- (\x-1,\y) -- (\x-1,\y+1) -- (\x-2,\y+1);
						\foreach \x in {1} \foreach \y in {\x+1,\x+4,\x+7,\x+10,\x+13,\x+16} \draw[ultra thick, ->, green] (\x,\y) -- (\x-1,\y);
						
						\foreach \x in {2,8,14} \foreach \y in {\x+1,\x+4,\x+7,\x+10,\x+13,\x+16} \draw[ultra thick, <-, black] (\x,\y) -- (\x-1,\y) -- (\x-1,\y+1) -- (\x-2,\y+1);
						
						\foreach \x in {0} \foreach \y in {2,5,8,11,14,17} \draw[ultra thick, ->, black] (\x,\y) -- (\x,\y+2);
						
						\foreach \y in {0} \foreach \x in {1,4,7,10,13,16} \draw[ultra thick, ->, red] (\x,\y) -- (\x+2,\y);

						%
						\foreach \x in {0,1,2,3,4,5,6,7,8,9,10,11,12,13,14,15,16,17} \coordinate [label=center:\x] (A) at (\x, -0.5);
						\foreach \y in {0,1,2,3,4,5,6,7,8,9,10,11,12,13,14,15,16,17}  \coordinate [label=center:\y] (A) at (-0.5,\y);
						
						\foreach \x in {0,1,2,3,4,5,6,7,8,9,10,11,12,13,14,15,16,17}
						\foreach \y in {0,1,2,3,4,5,6,7,8,9,10,11,12,13,14,15,16,17}
						{
							\filldraw[fill=white] (\x, \y) circle (0.1); 
						}
						
						\foreach \x in {0,1,2,3,4,5} 
						{
							\filldraw[fill=YellowOrange] (\x, \x) circle (0.1); 
						}
						\foreach \x in {0,1,2,3,4} 
						{
							\filldraw[fill=yellow] (\x, \x+1) circle (0.1); 
						}

						\end{tikzpicture}
					}	
				\end{center}
			\end{minipage}
		\end{center}
	\caption{Additional edges arising from conjuncts of group (4) (solid lines drawn inside grid cells) and transitivity (dashed lines). Nodes on the diagonals  are marked orange ($e$) and  yellow ($e'$).}
	\label{fig:embeddingRight}
\end{subfigure}
		\caption{Construction of the intended model of $\phi_{grid}$ 
			in the proof of Theorem~\ref{th:three}.}
		\label{fig:embedding}
	\end{figure}

	The formula $\phi_{grid}$ comprises a large number of conjuncts. To help give an overview of the construction, we have organized these conjuncts into groups, each of which secures a particular property (or collection of properties) exhibited by its models.
	The first two properties are very simple:
	\begin{enumerate}
		\item[(1)] There is an `initial' element satisfying $d_{00}(x) \wedge e(x) \wedge l(x)\wedge f(x)$.\label{rgb:initial}
		\item[(2)] The predicates $c_{i,j}$ and $d_{i,j}$ together partition the universe.
	\end{enumerate}	
	The third property generates the path shown in Fig.~\ref{fig:embeddingLeft}:
	\begin{enumerate}
		\item[(3)] Each element has a $b$- $r$- or $g$- successor as shown in the path shown in Fig.~\ref{fig:embeddingLeft}, and satisfying the appropriate predicates $c_{i,j}$ or $d_{i,j}$. Specifically, if a node in this path has coordinates $(x,y)$ with $y > x$,  	
		then it satisfies $c_{i,j}$ where $i = x \mod 3$ and $j = y \mod 6$; 
		and when $y \leq x$, then the node satisfies $d_{i,j}$ where $i = x \mod 3$ and $j = y \mod 6$.
	\end{enumerate}
	The conjuncts enforcing this property have the form
	\begin{equation}
	\forall x \big( colour(x)\wedge diag(x)\wedge border(x) \rightarrow \exists y (t(x,y)\wedge colour'(y))\big),\tag{3a}\label{rgb:witnesses}
	\end{equation}
	where $colour$ and $colour'$ stand for one of the predicate letters $c_{i,j}$ or $d_{i,j}$, $diag(x)$ stands for one of the literals $e'(x)$, $\neg e'(x)$, $e(x)$, $\neg e(x)$ or $\top$  (i.e.~the logical constant true), $border(x)$ stands for one of the literals $l(x)$, $\neg l(x)$, $f(x)$, $\neg f(x)$ or $\top$, and  $t$ stands for one of the transitive predicate letters $b$, $r$ or $g$. 
	The precise combinations of the literals and predicate letters in these conjuncts can be read from Fig.~\ref{fig:embeddingLeft} (cf.~Appendix, Table.~\ref{table:witnesses} for a full list). 
	%
	
To connect all pairs of elements that are neighbours in the standard grid we require a fourth property, which
we give in schematic form as follows:
	\begin{enumerate}
		
		\item[(4)]\label{IfBlack}
		{\em Certain} pairs of elements connected by \textit{one} transitive relation are also connected by \textit{another}, as indicated in Fig.~\ref{fig:embeddingRight}. 
\end{enumerate}		
Here are some examples of the conjuncts enforcing this property:
\begin{align}
		\forall x (c_{01}(x) &\rightarrow \forall y (b(x,y)\wedge (c_{11}(y) \vee d_{11}(y)) \rightarrow g(x,y))),\tag{4a}\label{tiling:4a}\\
		\forall x (d_{11}(x) &\rightarrow \forall y (b(x,y)\wedge d_{10}(y)\rightarrow r(x,y))),\tag{4b}\label{tiling:4b} \\
		\forall x (d_{11}(x) &\rightarrow \forall y (r(x,y)\wedge (c_{12}(y)\vee d_{12}(y))\rightarrow g(x,y))),\tag{4c}\label{tiling:4c}
		%
		\end{align}
		%
The role of these conjuncts can be explained referring to Fig.~\ref{fig:embeddingRight}.
For example, employing \eqref{tiling:4b} for the element $(1,1)$ in the intended model $\fG$, we get $\fG\models r((1,1),(1,0))$; hence by transitivity of $r$, also $\fG\models r((1,1),(1,2))$. This, applying \eqref{tiling:4c}, implies $\fG\models g((1,1),(1,2))$. By \eqref{tiling:4a}, we get $\fG \models g((0,1),(1,1))$ and, by transitivity of $g$, $\fG\models g((0,1),(0,2))$. The process is illustrated in Fig.~\ref{fig:embeddingRight}; when carried on along the zig-zag path, it constructs a grid-like structure. 

These conjuncts depend on having available the predicates marking the borders and the diagonals. Specifically, we require the following
property: 
	\begin{enumerate}
		%
		%
		\item[(5)] the predicates $l$, $f$, $e$ and $e'$ are distributed to mark the left-most column, the first row, the diagonal and
		the `super-diagonal'  
		of the grid, as indicated above. To secure this property, we add  to $\phi_{grid}$ several conjuncts, for instance:
	\end{enumerate}
	\begin{align}
	&\!\!\!\!\!\!\!\!\!	\bigwedge_{0\leq i\leq 2,\; 0\leq j\leq 5}
	\!\!\!\!\!\!\!\!\!\! \forall x \big(d_{i,j}(x)\wedge \pm e(x) \rightarrow \forall y ((b(x,y)\vee g(x,y) \vee r(x,y))\wedge d_{i+1,j+1}(y) \rightarrow \pm e(y)) \big),
	\tag{5a}\label{tiling:e}
	\end{align}
	where $\pm e(x)$ denotes uniformly $e(x)$ or $\neg e(x)$. Similar conjuncts are added for the super-diagonal, left column and bottom row; and also for the connection with and between $e$ and $e'$. 
The conjuncts ensuring properties (4) and (5) work in tandem.
For instance, applying~\eqref{tiling:e} to (1,1) we get $e$ is true at $(2,2)$; then, following the zig-zag path and applying more conjuncts from the group (4), we get that $g((2,2),(3,3))$ holds, so the node $(3,3)$ will be marked by $e$; this will propagate along the main diagonal. 
	
The structure $\fG$ depicted in Fig.~
\ref{fig:grid-three-transitive:Left}
 is a model of $\phi_{grid}$. In fact, $\phi_{grid}$ is an infinity axiom. To see this, let $\fA\models \phi_{grid}$ and define an injective embedding $\rho$ of the standard grid on $\N\times \N$ into $\fA$ as follows. Let $next:\N\times\N \mapsto \N\times \N$ be the successor function defined on $\N\times\N$ as depicted by the zig-zag path in the left-hand picture of Fig.~\ref{fig:embedding} starting at $(0,0)$ (ignoring any colours). Denote $s_0=(0,0)$, $s_n=next(s_{n-1})$ and $S_n=\{s_0, \ldots, s_n\}$. Let $a_0\in A$ be an element such that $\fA \models d_{00}(a)\wedge e(a) \wedge l(a)\wedge f(a)$ that exists by condition~(1). Define $\rho(s_0)=a_0$. Now, we proceed inductively: suppose $\rho(s_{n-1})$ has already been defined in step $n-1$ of the induction and $\rho(s_{n-1})=a_{n-1}$.
 Let $a_n$ be the witness of $a_{n-1}$ for the appropriate conjunct from the group~(3), i.e.~where the unary literals for $x$ agree with the unary literals satisfied by $a_{n-1}$ in $\fA$.  Define $\rho(s_n)=a_n$. 		
 Using induction one can  prove that $\rho$ is indeed injective: in the inductive step we assume that  
 $\fA\upharpoonright\{a_0, \ldots, a_{n-1}\}$ is isomorphic to $\fG\upharpoonright S_{n-1}$, and 
 we show that $a_n\not\in \{a_0, \ldots, a_{n-1}\}$ and 
 $\fA\upharpoonright\{a_0, \ldots, a_{n}\}$ and $\fG\upharpoonright S_n$ are again isomorphic. 
 In the proof one considers several cases depending on the 1-type realized by $a_n$. The formula $\phi_{grid}$ ensures that $a_0, \ldots, a_{18}$ are all distinct, and any eight consecutive elements of the sequence $a_0,\ldots, a_n$ are always distinct. Consider $a_{18}=\rho(4,2)$ that requires a witness $b\in A$ for a conjunct from the group~(3) such that $\fA\models r(a_{18},b)\wedge d_{11}(b)$. Suppose, $b=a_{2}=\rho(1,1)$, since $\fA\models d_{11}(a_2)$. Then, by transitivity of $g$, $\fA \models g(a_{18},a_{10})$, 
 which is a contradiction with $\fG \models \neg g((4,2),(1,1))$.  Other cases are similar and due to page limits have been omitted.

 We are now ready to define the horizontal and vertical successors in models of of $\phi_{grid}$. In fact, instead of defining the horizontal grid successor $\htw$ as one binary relation, we define two disjoint binary relations $\rt(x,y)$ and $\lt(x,y)$ such that $\rt$ and the inverse of $\lt$ together give the expected horizontal grid successor; they are defined respecting the `direction' of the transitive edges in the models. In the intended model  $\rt((x_1,y_1),(x_2,y_2))$ holds if $x_2=x_1+1$, $y_2=y_1$ and $(x_1,y_1)$ and $(x_2,y_2)$ are connected by $b$, $g$ or $r$; and for $\lt((x_1,y_1),(x_2,y_2))$ to hold we require $x_2=x_1-1$ instead of $x_2=x_1+1$. 
We present the definition of $\rt(x,y)$ in detail below\footnote{Addition in subscripts of the $c_{i,j}$'s is
always understood modulo 6 in the first position, and modulo 3 in the second position, i.e. $c_{i+a,j+b}$ denotes $c_{(i+a)\,\mathsf{mod}\: 6,(j+b)\,\mathsf{mod}\: 3}$.
Similarly, addition in subscripts of the $d_{i,j}$'s is understood modulo 3 in the first position, and modulo 6 in the second position.}
	\begin{align*}
	\rt(x,y) := &  (b(x,y) \vee g(x,y) \vee r(x,y))\quad \wedge  \\
	&	(c_{01}(x) \wedge d_{11}(y)) \vee 	(c_{20}(x) \wedge d_{03}(y)) \vee (c_{42}(x)\wedge d_{25}(y))\quad \vee\\
	& \mynegsp\mynegsp\mynegsp\mynegsp\mynegsp  \mynegsp\mynegsp\mynegsp\mynegsp 	\bigvee_{(i,j)\notin \{(0,2),(1,2),(2,1),(3,1),(4,0),(5,0)\} } \mynegsp\mynegsp\mynegsp \mynegsp\mynegsp\mynegsp\mynegsp   (c_{ij}(x) \wedge c_{i+1,j}(y) )\quad \vee
	\mynegsp	 \bigvee_{(i,j)\notin \{(2,1),(1,3),(0,5)\}}\mynegsp\mynegsp\mynegsp\mynegsp   (d_{ij}(x) \wedge d_{i+1,j}(y) )
	\end{align*}	
	The relation $\rt$ connects elements that are connected by  $b$, $g$ or $r$ and satisfy one of the possible combinations of colours: in the second line the combinations for crossing the diagonal are listed, in the third line the left disjunction describes combinations when both elements are located above the diagonal, and in the right disjunction---when both elements are located on and below the diagonal. The definition of  $\lt(x,y)$ complements that  of $\rt$.
	Analogously, we define relations $\up$ and $\dw$ that together define the vertical grid successor. 
	
	Now we are ready to write formulas that properly assign tiles to elements of the model. We do this with a formula  $\phi_{tile}$,
	which again features several conjuncts enforcing various properties of its models. Fortunately, the properties in question 
	are much simpler this time:
	\begin{enumerate}
		\item[(6)] Each node encodes precisely one tile.\label{tiling:C1}	
		\item[(7)] Adjacent tiles respect ${\cal C}_H$.
		\item[(8)] Adjacent tiles respect ${\cal C}_V$ (written as above using $\up$ and $\dw$). 
	\end{enumerate}
Property (6) is secured by the conjunct $\forall x \big(\bigvee_{C \in {\cal C}} C(x) \wedge \bigwedge_{C \neq D} (\neg C (x) \vee \neg D(x))\big)$.
Property (7)  is secured by the conjunct	
			\begin{eqnarray*}
				%
				\bigwedge_{C\in {\cal C}}	\forall x \big(C(x) \rightarrow \forall y \big((\rt(x,y) \rightarrow \bigvee_{C': (C,C') \in {\cal C}_H}  C'(y))
				\quad \wedge \quad 
				(\lt(x,y) \rightarrow \bigvee_{C': (C',C) \in {\cal C}_H}  C'(y))\big)\big); \label{three:H1f} 
			\end{eqnarray*}
and property (8) is analogous. We remark that these latter formulas are not strictly fluted but can be rewritten as fluted using classical tautologies (cf.~formula (7) in Section~\ref{app:twotrans}).
	
Finally, let $\eta_{\boldsymbol{\cal C}}$ be the conjunction of $\phi_{grid}$ and $\phi_{tile}$. The following Claim completes the reduction and the proof of our theorem.
	%
	\begin{claim}
		$\eta_{\boldsymbol{\cal C}}$ is satisfiable  iff  $\boldsymbol{\cal C}$ tiles  $\N \times \N$.
	\end{claim}
		\textit{Proof:} ($\Leftarrow$) If  $\boldsymbol{\cal C}$ tiles  $\N \times \N$ then  to show that $\eta_{\boldsymbol{\cal C}}$ is satisfiable we can expand our intended model $\fG$ for $\phi_{grid}$ assigning to every element of the grid a unique $C\in {\cal C}$ given by the tiling.  
		
		($\Rightarrow$) Let $\fA\models \eta_{\boldsymbol{\cal C}}$. 
		Let $\rho$ be the embedding of the standard $\N\times\N$ grid into $\fA$ defined above. One can inductively show that $\rho$  maps neighbours in the grid to elements connected by one of the relations $\lt, \rt, \up, \dw$ as follows ($i,j\geq 0$):
		$$\fA\models \rt(\rho(i,j),\rho(i+1,j)) \dot{\vee} \lt(\rho(i+1,j),\rho(i,j)) \wedge \up(\rho(i,j),\rho(i,j+1)) \dot{\vee} \dw(\rho(i,j+1),\rho(i,j)).$$
		(Here, $\dot{\vee}$ is exclusive disjunction.)
		So, we can define a tiling of the standard grid assigning to every node $(i,j)$ the unique tile $C$ such that $\fA \models C(\rho(i,j))$. Conditions (7) and (8) together with the above observation  ensure that this assignment satisfies the 
		tiling conditions.
\end{proof}
%
We remark that the formula $\phi_{grid}$ in the proof of Theorem~\ref{th:three} is an axiom of infinity, hence the satisfiability and the finite satisfiability problems do not coincide. Moreover, all formulas used in the proof are either guarded or can easily be rewritten as guarded. Furthermore, in the proof it would suffice to assume that $b$, $g$ and $r$ are interpreted as equivalence relations. Hence, we can strengthen the above theorem as follows.
\begin{corollary}
	The  satisfiability problem for the intersection of the fluted fragment with the two-variable guarded fragment is undecidable in the presence of three transitive relations \textup{(}or three equivalence relations\textup{)}. 
\end{corollary}

Now we improve the undecidability result to the case of $\FLtwotrans$ with  equality. 
\begin{theorem}
	The \textup{(}finite\textup{)} satisfiability problem for the two-variable fluted fragment with equality is undecidable in the presence of two transitive relations. \label{th:two}
\end{theorem}
\begin{proof}
	We write a formula $\phi_{grid}$ over a signature consisting of transitive relations $b$ and $r$, and unary predicates $c_{i,j}$ ($0\leq i,j\leq 3$).  The formula $\phi_{grid}$  captures several properties of the intended expansion of the $\Z\times\Z$ grid as shown  
	Fig.~\ref{fig:grid-three-transitive:Right}:
	%
	\begin{enumerate}[(1)]
		\item there is an initial element: $\exists x . c_{00}(x)$.
		\item the predicates $c_{i,j}$ partition the universe.
		\item transitive paths do not connect distinct elements of the same colour:
	$		\bigwedge_{0\leq i,j\leq 3} \forall x (c_{ij}(x)\rightarrow \forall y ((b(x,y)\vee r(x,y))\wedge c_{ij}(y) \rightarrow x=y))$
		%
		\item  each element belongs to a 4-element blue clique and to a 4-element red clique.
		\item {\em certain} pairs of elements connected by $r$ are also connected by $b$, and {\em certain} pairs of elements connected by $b$ are also connected by $r$.  
	\end{enumerate}		
    We have given property (5) only schematically, of course; its role is analogous to that of property (4) in the proof of Theorem~\ref{th:three}.
	The remainder of the proof is similar to the one presented for Theorem~\ref{th:three} and due to space limits it is relegated to the~Appendix. We note that $\phi_{grid}$ has also finite models expanding a toroidal grid structure $\Z_{4m}\times \Z_{4m}$  ($m>0$) obtained  by identifying elements from columns 0 and $4m$ and from rows 0 and $4m$. Hence, the proof gives undecidability for both the satisfiability and the finite satisfiability problems.
\end{proof}

Again, the formulas used in the above proof are guarded or can  be rewritten as guarded. Also it suffices to assume that $r$ is an equivalence relation. Hence we get the following
\begin{corollary}
	The \textup{(}finite\textup{)} satisfiability problem for the intersection of the fluted fragment with equality with the two-variable guarded fragment is undecidable in the presence of two transitive relations \textup{(}or one transitive and one equivalence relation\textup{)}.
\end{corollary}

\section{Conclusions}
In this paper, we considered the ($m$-variable) fluted fragment in the presence of different numbers of transitive relations.
We showed that $\FLotrans$ has the finite model property, but $\FLthreetrans$ admits axioms of infinity and the satisfiability problem for $\FLthreetrans$ is undecidable. This contrasts with known results for other decidable fragments, in particular, \FOt, where the satisfiability and finite satisfiability problems are undecidable
in the presence of two transitive relations, and where the \textit{finite} satisfiability problem is decidable in the presence of one transitive relation.
It is open whether the (finite) satisfiability problem for $\FL$ in the presence of \textit{two} transitive relations, $\ft_1$ and $\ft_2$, is decidable.
We point out that Lemma~\ref{lma:eliminate} in Section~\ref{sec:onetrans} could be generalized to normal form formulas from $\FLtwotransmpo$. Hence,  the (finite) satisfiability problem for  $\FL$ in the presence of two transitive relations is decidable if and only if the corresponding problem for $\FLt$
with two transitive relations is decidable. Unfortunately neither the method of Sec.~\ref{sec:onetrans} (to show
decidability) nor that of Sec.~\ref{sec:undecidable} (to show
undecidability) appears to apply here. The barrier in the former case is that pairs of elements can be related by both $\ft_1$ and $\ft_2$ via divergent $\ft_1$- and $\ft_2$-chains, so that simple certificates of the kind employed for $\FLotranstMinus$ do not guarantee the existence of models. The barrier in the latter case is that the grid construction has to build models featuring transitive paths of {\em bounded} length, and this seems not to be achievable with just two transitive relations. Finally, we expect that the undecidability result for $\FLthreetrans$ can be extended to get undecidability of the corresponding finite satisfiability problem.

\bibliography{purdyTrans} 

\appendix

\section{Proofs from Section \ref{sec:onetrans}}
\subsection{Proof of Lemma~\ref{lma:nf}}
We employ standard `re-writing' techniques, temporarily allowing 0-ary predicates. By moving negations inward in the usual way, we may assume without loss of generality that negation symbols in $\phi$ apply only to atoms. 
Suppose $\phi_0$ has a subformula $\theta= \exists x_{k+1} . \zeta$ ($0 \leq k < 2$), where $\zeta$ is quantifier-free. Let $p$ be a fresh predicate of arity $k$ (thus: 0-ary predicates can be introduced), and again let $\phi_1:= \phi_0[\theta/p(x_1, \dots, x_k)]$ and $\psi_1:= \forall x_1 \cdots \forall x_k (p(x_1, \dots, x_k) \leftrightarrow \theta)$; hence, $\models \phi_1 \wedge \psi_1 \rightarrow \phi_0$, and any model of $\phi_0$ can be expanded to a model of $\phi_1 \wedge \psi_1$. Thus, $\psi_1$ is of the form $\forall x_1 \cdots \forall x_k (p(x_1, \dots, x_k) \leftrightarrow \exists x_{k+1} . \zeta)$, where $\zeta$ is fluted and quantifier-free.
Suppose, alternatively, that $\phi_0$ has a subformula $\theta= \forall x_{k+1} . \zeta$. where $\zeta$ is quantifier-free. We proceed in the similarly, except that 
$\psi_1$ is of the form $\forall x_1 \cdots \forall x_k (p(x_1, \dots, x_k) \leftrightarrow \forall x_{k+1} . \zeta)$.

Proceeding similarly with $\phi_1$ in place of $\phi_0$, we obtain $\phi_2$ and $\psi_2$, and so on, until we reach some point where
the $\FL^{m}$-sentence $\phi_s$ is a proposition letter. Defining $\psi'$ to be $\psi_1 \wedge \cdots \wedge \psi_s \wedge \phi_s$, we see that 
$\models \psi' \rightarrow \phi$, and any model of $\phi$ can be expanded to a model of $\psi'$. It should be clear that the size
of $\psi'$ is at most linear in the size of $\phi$. Each of the $\psi_i$ and indeed $\phi_s$
is either a proposition letter or (following variable re-naming) of one of the forms 
\begin{align*}
& \ \forall x_1  (p(x_1) \leftrightarrow \exists x_2 . \lambda) \qquad \forall x_1  (p(x_1) \leftrightarrow \forall x_2 . \lambda)
\qquad  P \leftrightarrow \forall x_1 . \mu\qquad P \leftrightarrow \exists x_1. \mu,
\end{align*} 
where $\lambda$ is a quantifier-free fluted formula in variables $x_1, x_2$ and $\mu$ a quantifier-free fluted formula with variable $x_1$. Furthermore, we may re-write any formula $\forall x_1  (p(x_1) \leftrightarrow \exists x_2 . \lambda)$ equivalently as a conjunction
$\forall x_1  (p(x_1) \rightarrow \exists x_2 . \lambda) \wedge \forall x_1  (\neg p(x_1) \rightarrow \forall x_2  \neg \lambda)$, 
and similarly for $\forall x_1  (p(x_1) \leftrightarrow \forall x_2 . \lambda)$.

Eliminate the 0-ary predicates by guessing their truth values and carrying out the obvious simplifications. Let $\Psi'$ be any
satisfiable collection
of formulas obtained in this way if there is one (or any of them if $\phi$ is not satisfiable).
Massaging $\Psi'$ into a set of formulas $\Psi$ of the desired forms is then completely routine.

\subsection{Proof of Lemma~\ref{lma:resolution}}\label{app:resolution}
	Enumerate the atoms $r(x_1, \dots, x_m)$ where  $r \in \Sigma$ is of maximal arity $m$, as 
$\rho_1, \dots, \rho_n$.
Define a \textit{level-$i$ extension}
of $\tau^-$ inductively as follows: (i) $\tau^-$ is a level-0 extension of $\tau^-$; (ii) if $\tau'$ is a level-$i$ extension of $\tau^-$ ($0 \leq i < n$), then $\tau' \wedge \rho_{i+1}$  and $\tau' \wedge \neg \rho_{i+1}$ are 
level-$(i+1)$ extensions of $\tau^-$. Thus, the level-$n$ extensions of $\tau^-$ are exactly 
the fluted $m$-types over $\Sigma \cup \set{\ft}$ entailing
$\tau^-$. 
If $\tau'$ is a level-$i$ extension of $\tau^-$ ($0 \leq i < n$), we say that $\tau'$ {\em violates} a clause $\delta$
if, for every literal in $\delta$, the opposite literal is in $\tau'$; we say that $\tau'$ {\em violates} a set of clauses
$\Delta$ if $\tau'$ violates some $\delta \in \Delta$.
We construct a sequence of level-$i$ extensions of $\tau^-$ ($i = 0, 1, \dots$) none of which violates $\Gamma^*$.

By definition, $\tau^-$ is a level-0 extension of itself. 
Suppose now that $\tau'$ is a level-$i$ extension of $\tau^-$ ($0 \leq i < n$). We claim that, if both 
$\tau' \wedge \rho_{i+1}$  and $\tau' \wedge \neg \rho_{i+1}$ violate $\Gamma^*$, then so does $\tau'$. For otherwise, there must be a clause $\neg \rho_{i+1} \vee \gamma' \in \Gamma^*$ violated by $\tau' \wedge \rho_{i+1}$ and a clause $\rho_{i+1} \vee \delta' \in \Gamma^*$ violated by $\tau' \wedge \neg \rho_{i+1}$. But in that case $\tau'$ violates the fluted resolvent $\gamma' \vee \delta'$,  contradicting the
supposition that $\tau'$ does not violate $\Gamma^*$. This proves the claim. Now, since $\tau^-$ 
by hypothesis is consistent with $\Gamma^\circ$, it does not violate $\Gamma^\circ$. Therefore, since $\tau^-$ involves no $m$-literals
$\pm r(x_1, \dots, x_m)$ for $r \in \Sigma$, it does not violate $\Gamma^*$ either. By the above claim, then,
there must be at least one level-$n$ extension $\tau$ of $\tau^-$ which does not
violate $\Gamma^* \supseteq \Gamma$. Since $\tau$ is a fluted $m$-type, this proves the lemma.

\subsection{Proof of Lemma~\ref{lma:eliminateNonTrans}}\label{app:eliminateNonTrans}
Let $\phi$ be the formula
%
\begin{equation}
\forall x_1 x_{2} . \Omega \ \wedge
{\bigwedge_{i=1}^{s} }\forall x_1 \left(p_i(x_1) \rightarrow \exists x_2 . \Gamma_i\right) \wedge
{\bigwedge_{j=1}^{t}} \forall x_1 (q_j(x_1) \rightarrow \forall x_2 . \Delta_j),
\label{eq:cnfTwo}
\end{equation}
where $\Omega, \Gamma_1, \dots, \Gamma_s, \Delta_1, \dots, \Delta_t$ are sets of fluted 2-clauses,
and $p_1, \dots, p_s$, $q_1, \dots, q_t$ unary predicates. Write $T = \{1, \dots, t\}$.
For all $i$ ($1 \leq i \leq s$) and all $J \subseteq T$, let $p_{i,J}$ and $q_J$ be new unary
predicates.
The intended interpretation of $p_{i,J}(x_1)$ is ``$x_1$ satisfies $p_i$, and also satisfies $q_j$ for every $j \in J$;'' and
the intended interpretation of $q_{J}(x_1)$ is ``$x_1$ satisfies $q_j$ for every $j \in J$.''
Let $\phi'$ be the conjunction of the sentences
\begin{align}
&
\bigwedge_{i=1}^{s} \bigwedge_{J \subseteq T} \forall x_2
(p_i(x_2) \wedge \bigwedge_{j \in J} q_j(x_2))  \rightarrow p_{i,J}(x_2))
\label{eq:phiStar1Two}\\
&
\bigwedge_{J \subseteq T}\forall x_2
((\bigwedge_{j \in J} q_j(x_2)) \rightarrow q_{J}(x_2))
\label{eq:phiStar2Two}\\
& 
\bigwedge_{i=1}^{s} \bigwedge_{J \subseteq T} \forall x_1
\left( p_{i,J}(x_1)\rightarrow \exists x_2 \left(\Gamma_i \cup \Omega \cup \bigcup \{\Delta_j \mid j \in J \}\right)^\circ \right) 
\label{eq:phiStar3Two}\\ 
&
\bigwedge_{J \subseteq T}\forall x_1 
\left(q_{J}(x_1) \rightarrow \forall x_2 \left(\Omega \cup \bigcup \{\Delta_j \mid j \in J \}\right)^\circ \right).
\label{eq:phiStar4Two}
\end{align}
We claim that, if $\phi$ is satisfiable, then so is $\phi'$ (see Appendix, Section~\ref{app:eliminateNonTrans}). This proves the lemma, since $\phi'$ evidently contains no
non-distinguished binary predicates, and hence is in $\FLotranstMinus$.

For suppose $\fA \models \phi$. We expand
$\fA$ to a model $\fA' \models \phi'$ by setting, for all $i$ ($1 \leq i \leq s$) and all $J \subseteq T$,\\
$	p_{i,J}^{\fA'}  = \{ a \mid \text{$\fA \models p_i[a]$ and $\fA \models q_j[a]$ for all $j \in J$}  \}$ and 
$q_{J}^{\fA'}  = \{ a \mid \text{$\fA \models q_j[a]$ for all $j \in J$}  \}$.

\medskip\noindent	To see that $\fA' \models \phi'$, we simply check the truth of conjuncts~\eqref{eq:phiStar1Two}--\eqref{eq:phiStar4Two} in $\fA'$ in turn.
Sentences~\eqref{eq:phiStar1Two} and~\eqref{eq:phiStar2Two} are immediate. For~\eqref{eq:phiStar3Two}, fix $i$ and $J$, and suppose
$\fA' \models p_{i,J}[a_1]$. By the definition of $\fA'$, $\fA \models p_i[a_1]$ and 
$\fA \models q_j[a_1]$ for all $j \in J$. Since $\fA \models \phi$, there exists
$b$ such that $\fA \models \Gamma_i[a_1, b]$, \mbox{$\fA \models \Omega[a_1, b]$} and
$\fA \models \Delta_j[a_1, b]$ for all $j \in J$. Since resolution is a valid inference step,
$\fA' \models \left(\Gamma_i \cup \Omega \cup \bigcup \{\Delta_j \mid j \in J \}\right)^\circ[a_1,b]$. This establishes the truth of~\eqref{eq:phiStar3Two} in $\fA'$. Sentence~\eqref{eq:phiStar4Two} is handled similarly.

Conversely, we claim that, if 
$\phi'$ is satisfiable over a domain $A$, then $\phi$  is satisfiable over a domain of size  $s \cdot |A|$.
For suppose $\fA \models \phi'$. Let $\fB$ be the model of $\phi'$ guaranteed by Lemma~\ref{lma:multiply}, where $z= s$. 
We may assume that $\fA$ and hence $\fB$ interpret no non-distinguished predicates of arity $2$.
We proceed to expand $\fB$ to a model $\fB' \models \phi$ by interpreting the non-distinguished predicates of arity $2$ occurring in $\phi$.      
Pick any element $a_1$ from $B$, and let $J$ be the set of all $j$ ($1 \leq j \leq t$)
such that $\fB \models q_j[a_1]$. Suppose also that, for some $i$ ($1 \leq i \leq s$)
$\fA \models p_i[a_1]$. 
From~\eqref{eq:phiStar1Two}, $\fB \models p_{i,J}[a_1]$; and
from~\eqref{eq:phiStar3Two},
we may pick $b_i \in B$ such that
$\fB \models \left(\Gamma_i \cup \Omega \cup \bigcup \{\Delta_j \mid j \in J \}\right)^\circ[a_1,b_i]$. From the properties secured for $\fB$ by Lemma~\ref{lma:multiply}, we know that if, for fixed $a$,
we have $\fA \models \alpha_i[a_1]$ for more than one value of $i$, then
we may choose the corresponding elements $b_{i}$ so that they are all distinct.  
For each such $b_i$, then, let $\tau_i = \ftp^{\fB}[a_1,b_{i}]$. Thus, $\tau_i(x_1,x_2)$ is consistent
with $\left(\Gamma_i \cup \Omega \cup \bigcup \{\Delta_j \mid j \in J \}\right)^\circ$. By Lemma~\ref{lma:resolution}, there exists a 
fluted $2$-type $\tau_i^+ \supseteq \tau_i(x_1, x_2)$ such that $\tau_i^+$ is consistent with $\Gamma_i \cup \Omega \cup  \bigcup \{\Delta_j \mid j \in J \}$. 
Set $\ftp^{\fB'}[a_1,b_{i}] = \tau^+_i$. Since $\tau^+_i \supseteq \tau_i$,
only non-distinguished binary predicates are being assigned, so that there is no clash with $\fB$. Moreover, since the $b_i$ are all distinct (for a given element $a$),
these assignments do not clash with each other. In this way, every existential conjunct of $\phi$ is witnessed in $\fB'$ for every element $a$, and no static or universal conjunct of $\phi$ is violated for the tuples from $B$ for which the binary predicates of $\Sigma$ have been defined.
Now let $\langle a_1, a_2 \rangle$ be any ordered pair from $B$ for which the binary predicates
of $\Sigma$ have not been defined, and let $J$ be the set of all $j$ ($1 \leq j \leq t$)
such that $\fB \models q_j[a_1]$. From~\eqref{eq:phiStar2Two},
$\fB \models q_J[a_1]$; and from~\eqref{eq:phiStar4Two}, 
$\fB \models \left(\Omega \cup  \bigcup \{\Delta_j \mid j \in J \} \right)^\circ[a_1,a_2]$. 
Let $\tau= \ftp^\fB[a_1,a_2]$.  Hence $\tau(x_1, x_2)$ is consistent with 
$\left(\Omega \cup \bigcup \{\Delta_j \mid j \in J \} \right)^\circ$. By Lemma~\ref{lma:resolution}, there exists a 
fluted 2-type $\tau^+ \supseteq \tau$ such that $\tau^+$ is consistent 
with $\Omega \cup  \bigcup \{\Delta_j \mid j \in J \}$. Set $\ftp^{\fB'}[a_1,a_2] = \tau^+$. 
Since $\tau^+ \supseteq \tau$, only non-distinguished binary predicates are being assigned, so that there is no clash with $\fB$. 
Evidently, no static or universal conjunct of $\phi$ is violated in this process. Thus, $\fB' \models \phi$, as required.

\subsection{Proof of Lemma~\ref{lma:eliminate}}
	Similar
	to the proof of Lemma~\ref{lma:eliminateNonTrans}. Taking $\phi$ to be as in~\eqref{eq:cnf},
	write $T = \{1, \dots, t\}$. For
	all $i$ ($1 \leq i \leq s$) and all $J \subseteq T$, let $p_{i,J}$ and $q_J$ be new
	predicates of arity $m-2$.
	The intended interpretation of $p_{i,J}(x_2, \dots, x_{m-1})$ is ``for some $x_1$, 
	the tuple $x_1, \dots, x_{m-1}$ satisfies $\alpha_i$ and also satisfies $\beta_j$ for every $j \in J$;'' and
	the intended interpretation of $q_{J}(x_2, \dots, x_{m-1})$ is ``for some $x_1$, 
	the tuple $x_1, \dots, x_{m-1}$ satisfies $\beta_j$ for every $j \in J$.''
	Let $\phi'$ be the conjunction of the sentences
	\begin{align}
	&
	\bigwedge_{i=1}^{s} \bigwedge_{J \subseteq T} \forall x_1 \cdots \forall x_{m-1}
	((\alpha_i \wedge \bigwedge_{j \in J} \beta_j )  \rightarrow p_{i,J}(x_2, \dots, x_{m-1}))
	\label{eq:phiStar1}\\
	&
	\bigwedge_{J \subseteq T}\forall x_1 \cdots \forall x_{m-1}
	((\bigwedge_{j \in J} \beta_j) \rightarrow q_{J}(x_2, \dots, x_{m-1}))
	\label{eq:phiStar2}\\
	& 
	\bigwedge_{i=1}^{s} \bigwedge_{J \subseteq T} \forall x_2 \cdots \forall x_{m-1}
	( p_{i,J}(x_2, \dots, x_{m-1})\rightarrow \exists x_m (\Gamma_i \cup \Omega \cup \bigcup_{j \in J} \Delta_j )^\circ )
	\label{eq:phiStar3}\\ 
	&
	\bigwedge_{J \subseteq T}\forall x_2 \cdots \forall x_{m-1}
	( q_{J}(x_2, \dots, x_{m-1})\rightarrow \forall x_m (\Omega \cup  \bigcup_{j \in J} \Delta_j )^\circ ).
	\label{eq:phiStar4}
	\end{align}
	We claim that, if $\phi$ is satisfiable, then so is $\phi'$. Note that, since $m \geq 2$, \eqref{eq:phiStar3} and~\eqref{eq:phiStar4} do not involve $x_1$. By decrementing all variable indices in these conjuncts, therefore,
	we obtain a formula of $\FLotransm$ as required by the lemma. The proof of the claim proceeds almost identically to Lemma~\ref{lma:eliminateNonTrans}.

\section{Proof of Theorem~\ref{th:two}: undecidability of $\FLtwotrans$ with equality}\label{app:twotrans}
Below we present the complete proof that have been roughly sketched in Section~\ref{sec:undecidable}. 

Suppose the signature contains two transitive relations $b$ (blue) and $r$ (red), and additional unary predicates $c_{i,j}$ ($0\leq i\leq 3$, $0\leq j\leq 3$) called {\em colours}. We write a formula $\phi_{grid}$  capturing several properties of the intended expansion of the $\Z\times\Z$ grid  as shown in Fig.~\ref{fig:grid-two-transitive}.  There, each element with coordinates $(k,l)$ satisfies $c_{i,j}$, where $i=k\mod 4$ and $j=l\mod 4$ and the transitive relations connect only some elements that are close in  the grid.  The formula $\phi_{grid}$ is a conjunction of the following statements. 

\begin{figure}[thb]
	\begin{center}
		\begin{minipage}{13cm}
			\begin{minipage}{12.5cm}
				\begin{center}
					\begin{tikzpicture}[xscale=1,yscale=1]
					\clip (-0.15,-0.15) rectangle (7.3,4.3);
					\foreach \x in {-1,1,3,5,7}
					\foreach \y in {-1,1,3,5,7} \foreach \z in {0.1} 
					{
						\draw[color=red,thick,-, >=latex] (\x+\z, \y+\z) -- (\x+1-\z, \y+\z) --
						(\x+1-\z, \y+1-\z) -- (\x+\z, \y+1-\z) -- (\x+\z, \y+\z);
						\draw[color=red,thick,-, >=latex] (\x+\z, \y+\z) --
						(\x+1-\z, \y+1-\z) -- (\x+\z, \y+1-\z) -- (\x+1-\z, \y+\z);
						
					}	
					
					\foreach \x in {0,2,4,6}
					\foreach \y in {0,2,4,6} \foreach \z in {0.1} 
					{
						\draw[color=blue,thick,-, >=latex] (\x+\z, \y+\z) -- (\x+1-\z, \y+\z) --
						(\x+1-\z, \y+1-\z) -- (\x+\z, \y+1-\z) -- (\x+\z, \y+\z);
						\draw[color=blue,thick,-, >=latex] (\x+\z, \y+\z) --
						(\x+1-\z, \y+1-\z) -- (\x+\z, \y+1-\z) -- (\x+1-\z, \y+\z);
					}	
					
					\foreach \x in {1,5}
					\foreach \y in {1,5} \foreach \z in {0.15} 
					{
						\draw[color=blue,thick,->, >=latex] (\x, \y) -- (\x+1-\z, \y);
						\draw[color=blue,thick,->, >=latex] (\x, \y) -- (\x, \y+1-\z);
						\draw[color=blue,thick,->, >=latex] (\x-1, \y) -- (\x-1, \y+1-\z);
						\draw[color=blue,thick,->, >=latex] (\x, \y-1) -- (\x+1-\z, \y-1);						
					}	
					
					\foreach \x in {3,7}
					\foreach \y in {3,7} \foreach \z in {0.15} 
					{
						\draw[color=blue,thick,->, >=latex] (\x, \y) -- (\x+1-\z, \y);
						\draw[color=blue,thick,->, >=latex] (\x, \y) -- (\x, \y+1-\z);
						\draw[color=blue,thick,->, >=latex] (\x-1, \y) -- (\x-1, \y+1-\z);
						\draw[color=blue,thick,->, >=latex] (\x, \y-1) -- (\x+1-\z, \y-1);						
					}

					\foreach \x in {1,5}
					\foreach \y in {3,7} \foreach \z in {0.15} 
					{
						\draw[color=blue,thick,->, >=latex] (\x+1-\z, \y) -- (\x, \y) ;
						\draw[color=blue,thick,->, >=latex] (\x, \y+1-\z) -- (\x, \y);
						\draw[color=blue,thick,->, >=latex] (\x-1, \y+1-\z) -- (\x-1, \y);
						\draw[color=blue,thick,->, >=latex] (\x+1-\z, \y-1) -- (\x, \y-1);						
					}	
					
					\foreach \x in {3,7}
					\foreach \y in {1,5} \foreach \z in {0.15} 
					{
						\draw[color=blue,thick,->, >=latex] (\x+1-\z, \y) -- (\x+\z, \y) ;
						\draw[color=blue,thick,->, >=latex] (\x, \y+1-\z) -- (\x, \y+\z);
						\draw[color=blue,thick,->, >=latex] (\x-1, \y+1-\z) -- (\x-1, \y+\z);
						\draw[color=blue,thick,->, >=latex] (\x+1-\z, \y-1) -- (\x+\z, \y-1);						
					}	
					
					\foreach \x in {1,5}
					\foreach \y in {3,7} \foreach \z in {0.15} 
					{
						\draw[color=blue,thick,->, >=latex] (\x+1-\z, \y) -- (\x+\z, \y) ;
						\draw[color=blue,thick,->, >=latex] (\x, \y+1-\z) -- (\x, \y+\z);
						\draw[color=blue,thick,->, >=latex] (\x-1, \y+1-\z) -- (\x-1, \y+\z);
						\draw[color=blue,thick,->, >=latex] (\x+1-\z, \y-1) -- (\x+\z, \y-1);						
					}

					\foreach \x in {0,1,2,3,4,5,6,7}
					\foreach \y in {0,1,2,3,4,5,6,7}
					{
						\filldraw[fill=white] (\x, \y) circle (0.2); 
						\pgfmathtruncatemacro{\abc}{mod(\y,4)} 
						\pgfmathtruncatemacro{\bcd}{mod(\x,4)} 
						\coordinate [label=center:$_{\bcd\abc}$] (A) at (\x,\y);
					}

					
					\draw[ ->, dotted, very thick] (7.5,0) -- (8,0);
					\draw[ ->, dotted, very thick] (0,7.5) -- (0,8);
					\draw[ ->, dotted, very thick] (7.5,7.5) -- (8,8);
					
					%
					\end{tikzpicture}\end{center}
			\end{minipage}
			
		\end{minipage}		
	\end{center}
	
	\caption{Intended expansion of the $\N\times\N$ grid with two transitive relations {\color{blue}$b$}   and  {\color{red}$r$}. Edges without arrows represent connections in both direction. Nodes are marked by the indices of  the $c_{ij}$s they satisfy.  
		\label{fig:grid-two-transitive}}
\end{figure}
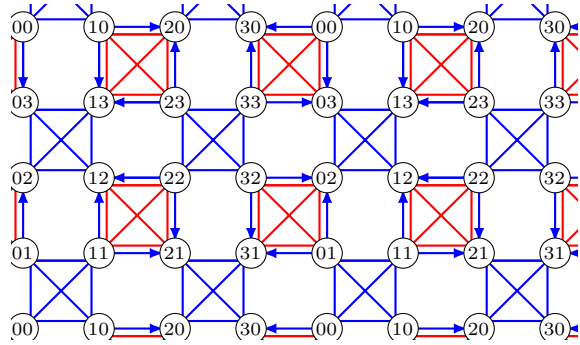

\begin{enumerate}[(1)]
	\item there is an initial element: $\exists x c_{00}(x) $.

	\item the predicates $c_{i,j}$ enforce a partition of the universe: $	\forall x \dot{\bigvee}_{0\leq i \leq 3} \dot{\bigvee}_{0\leq j \leq 3} c_{i,j}(x) 
	$.	\label{tiling:partition}
	\item transitive paths do not connect distinct elements of the same colour:
	\begin{eqnarray*}
		\bigwedge_{0\leq i,j\leq 3} \forall x (c_{ij}(x)\rightarrow \forall y ((b(x,y)\vee r(x,y))\wedge c_{ij}(y) \rightarrow x=y))
	\end{eqnarray*}\label{Clique}
	
	\item   
	-each element belongs to a 4-element blue clique; this property is expressed by writing the following conjuncts for each $i,j\in \{0,2\}$:\label{CliqueBlue}
	\begin{align*}
	%
	\forall x (c_{ij}(x) &\rightarrow \exists y (b(x,y)\wedge c_{i+1,j}(y)))\\ 
	\forall x (c_{i+1,j}(x) &\rightarrow \exists y (b(x,y)\wedge c_{i+1,j+1}(y))) 	\nonumber\\
	\forall x (c_{i+1,j+1}(x) &\rightarrow \exists y (b(x,y)\wedge c_{i,j+1}(y))) \nonumber\\
	\forall x (c_{i,j+1}(x) &\rightarrow \exists y (b(x,y)\wedge c_{ij}(y))) 
	\end{align*}
	%
	
	-each element belongs to a 4-element red clique;  we write the following conjuncts for each $i,j\in \{1,3\}$:\label{CliqueRed}
	%
	\begin{align*}
	%
	\forall x (c_{ij}(x) &\rightarrow \exists y (r(x,y)\wedge c_{i+1,j}(y))) \\
	\forall x (c_{i+1,j}(x) &\rightarrow \exists y (r(x,y)\wedge c_{i+1,j+1}(y))) \nonumber\\
	\forall x (c_{i+1,j+1}(x) &\rightarrow \exists y (r(x,y)\wedge c_{i,j+1}(y)))  \\
	\forall x (c_{i,j+1}(x) &\rightarrow \exists y (r(x,y)\wedge c_{ij}(y)))
	\end{align*}
	
	\item 		- a group of formulas saying that some pairs of elements connected by $r$ are also connected by $b$: 
	\begin{align}
	\bigwedge_{i=0,2} \forall x (c_{ii}(x) \rightarrow &\forall y (r(x,y)\wedge (c_{i,i-1}(y)\vee c_{i-1,i}(y) )\rightarrow b(x,y)))\tag{6a}
	\label{IfRed00}\\
	\bigwedge_{i=1,3} \forall x (c_{ii}(x) \rightarrow &\forall y (r(x,y)\wedge (c_{i,i+1}(y)\vee c_{i+1,i}(y))\rightarrow b(x,y)))\tag{6b}
	\label{IfRed11}\\	
	\bigwedge_{i=0,2} \forall x (c_{i,i+1}(x) \rightarrow &\forall y (r(x,y)\wedge (c_{i,i+2}(y)\vee c_{i-1,i+1}(y))\rightarrow b(x,y)))\tag{6c}
	\label{IfRed01}\\
	\bigwedge_{i=1,3} \forall x (c_{i,i-1}(x) \rightarrow &\forall y (r(x,y)\wedge (c_{ii}(y)\vee c_{i,i-2}(y) )\rightarrow b(x,y)))\tag{6d}
	\label{IfRed10}
	\end{align}
	
	- a group of formulas saying that some pairs of elements connected by $b$ are also connected by $r$: 
	\begin{align}
	\bigwedge_{i=0,2} \forall x (c_{ii}(x) \rightarrow &\forall y (b(x,y)\wedge (c_{i,i-1}(y)\vee c_{i-1,i}(y) )\rightarrow r(x,y))) \tag{7a}
	\label{IfBlue00}\\
	\bigwedge_{i=1,3} \forall x (c_{ii}(x) \rightarrow &\forall y (b(x,y)\wedge (c_{i,i+1}(y) \vee  c_{i+1,i}(y) )\rightarrow r(x,y)))\tag{7b}
	\label{IfBlue11}\\
	\bigwedge_{i=0,2} \forall x (c_{i,i+1}(x) \rightarrow &\forall y (b(x,y)\wedge (c_{i,i+2}(y)\vee c_{i-1,i+1}(y))\rightarrow r(x,y))) \tag{7c}
	\label{IfBlue01}\\
	\bigwedge_{i=1,3} \forall x (c_{i,i-1}(x) \rightarrow &\forall y (b(x,y)\wedge (c_{ii}(y)\vee c_{i,i-2}(y))\rightarrow r(x,y)))\tag{7d}
	\label{IfBlue10}
	\end{align}
\end{enumerate}

One can note that $\phi_{grid}$ has also finite models expanding a toroidal grid structure $\Z_{4m}\times \Z_{4m}$  ($m>0$) obtained  by identifying elements from columns 0 and $4m$ and from rows 0 and $4m$. 

In order to encode tilings using two-variable logics it actually suffices  to define structures that are grid-like. 
A structure $\fG=(G,h,v)$ with two binary relation $h$ and $v$ is {\em grid-like}, if one of the standard grids $\N\times \N$, $\Z\times\Z$ or $\Z_t\times\Z_t$ can be homomorphically embedded into $\fG$. To show that a structure $\fG$ is grid-like it suffices to require that $\fG \models \forall x (\exists y h(x,y) \wedge \exists y v(x,y))$ (note that this is an $\FLt$-formula) and the following {\em confluence} property holds
\begin{equation*}
\fG \models  \forall x,x',y,y' ((h(x,y)\wedge v(x,x')\wedge v(y,y') \rightarrow h(x',y')).\tag{*}
\end{equation*}\label{tiling}
The confluence property above uses four variables and is not fluted. We will enforce it in $\FLt{}$ using the transitive relations.
Let us introduce the following definitions:
\begin{align*}
\htw(x,y) := & b(x,y) \wedge \mynegsp
\bigvee_{{\substack{i=0,2\\ j=0,1,2,3 }}} \mynegsp
(c_{ij}(x) \wedge c_{i+1,j}(y) )\vee  r(x,y)  \wedge \mynegsp
\bigvee_{\substack{i=1,3\\ j=0,1,2,3 }} \mynegsp
(c_{ij}(x) \wedge c_{i+1,j}(y) )\\
\vtw(x,y) := &  b(x,y)  \wedge \mynegsp
\bigvee_{\substack{i=0,1,2,3\\ j=0,2 }}\mynegsp (c_{ij}(x) \wedge c_{i,j+1}(y) )\vee 
r(x,y) \wedge  \mynegsp
\bigvee_{\substack{i=0,1,2,3\\ j=1,3 }}\mynegsp (c_{ij}(x) \wedge c_{i,j+1}(y) )
\end{align*}	
Let $\fA\models \phi_{grid}$. 	We show that every model $(\fA,\htw,\vtw)$ of $\phi_{grid}$ is grid-like.  
We first show that $\fA$	 satisfies $\forall x (\exists y \htw(x,y) \wedge \exists y \vtw(x,y))$. One considers several cases depending on the values of the unary predicates.

Let $a\in A$ and assume $\fA\models c_{00}(a)$. By \eqref{CliqueBlue} there are $a_1,a_2,a_3,a_4\in A$ such that $\fA \models b(a,a_1)\wedge c_{10}(a_1) \wedge b(a_1,a_2)\wedge c_{11}(a_2)\wedge b(a_2,a_3)\wedge c_{01}(a_3) \wedge b(a_3,a_4)\wedge c_{00}(a_4)$. 
By \eqref{Clique} $a=a_4$ and by transitivity of $b$ the elements $a,a_1,a_2,a_3$ form a blue clique in $\fA$. Hence, $\fA\models \htw(a,a_1)\wedge \vtw(a,a_3)$. 

The same argument works if $a$ has one of the colours $c_{02}, c_{20}$ or $c_{22}$ and, similarly, applying (5) 
instead of \eqref{CliqueBlue} when $a$ has the colours $c_{11},c_{31},c_{13}$ or $c_{33}$. 

Consider now the case $\fA\models c_{10}(a)$. By \eqref{CliqueBlue} there is $a'\in A$ such that $\fA\models b(a,a')\wedge c_{11}(a')$, hence $\fA\models \vtw(a,a')$. Moreover, by \eqref{CliqueRed} there are $a_1,a_2,a_3,a_4\in A$ such that $\fA \models r(a,a_1)\wedge c_{13}(a_1) \wedge r(a_1,a_2)\wedge c_{23}(a_2)\wedge r(a_2,a_3)\wedge c_{20}(a_3)\wedge  r(a_3,a_4)\wedge c_{10}(a_4)$. 
By \eqref{Clique} $a=a_4$ and by transitivity of $r$ the elements $a,a_1,a_2,a_3$ form a red clique in $\fA$. By \eqref{IfRed10} $\fA\models b(a,a_1)$, hence $\fA\models \htw(a,a_1)$. 

Remaining cases are shown similarly. 

Now,  we show the confluence property (*). Let $a,a',b,b'\in A$ and $\fA \models  \htw(a,b)\wedge \vtw(a,a')\wedge \htw(b,b')$. One needs to consider several cases depending on the colour of $a$, in each of them showing that $\fA \models \htw(a',b')$. For instance:
\begin{itemize}
	\item	$\fA \models c_{00}(a)$. Then 
	$\fA \models c_{01}(a')\wedge b(a,a')\wedge b(a,b) \wedge c_{10}(b)\wedge b(b,b') \wedge c_{11}(b') $. By \eqref{CliqueBlue} $b'$ is a member of a blue clique containing elements of colours $c_{11}, c_{01}, c_{00}, c_{10}$. Since by \eqref{Clique} the relation $b$ does not connect distinct elements of the same colour, $a'$ belongs to the blue clique of $b'$ and $\fA \models b(a',a)$. Now, by transitivity of $b$, $\fA \models \htw(a',b')$. 
	%
	
	\item 
	$\fA \models c_{30}(a)$. Then  $\fA \models c_{31}(a')\wedge b(a,a')\wedge r(a,b) \wedge c_{00}(b)\wedge b(b,b') \wedge c_{01}(b') $. 
	Similarly as above, $b$ belongs to a red clique of $a$, hence $\fA\models r(b,a)$. 
	By \eqref{IfRed00}, $\fA \models b(b,a)$. Moreover, $b'$ is in a blue clique of $b$, and so $\fA \models b(b',b)$. 
	By transitivity of $b$,  $\fA \models b(b',a')$.  Now, by \eqref{IfBlue01}, $\fA\models r(b',a')$. By \eqref{CliqueRed}, $a'$ is a member of a red clique that, by \eqref{Clique}, must contain $b'$. Hence $\htw(a',b')$ holds.
	
	
\end{itemize}	
Remaining cases can be shown in a similar way. Hence, every model of $\phi_{grid}$ is grid-like. Now we ensure that we also can assign tiles to elements of the grid-like models using fluted formulas. The task in $\FOt$ is easy, it suffices to require that
\begin{enumerate}
	\item[(6)] each node encodes precisely one tile: \quad $	\forall x (\dot{\bigvee}_{C \in {\cal C}} C x)$,\label{tiling:C1}
	\item[(7)] adjacent tiles respect ${\cal C}_H$ and ${\cal C}_V$:
\end{enumerate}
\begin{align*}
\bigwedge_{C\in {\cal C}}	\forall x (Cx \rightarrow \forall y (\htw(x,y) \rightarrow \bigvee_{C': (C,C') \in {\cal C}_H}  C'y))\tag{7a}\label{tiling:H1fE}\\
\bigwedge_{C\in {\cal C}}	\forall x (Cx \rightarrow \forall y (\vtw(x,y) \rightarrow \bigvee_{C':(C,C') \in {\cal C}_V} C'y)).\tag{7b}\label{tiling:V1fE}
\end{align*}
The above two formulas are not fluted but can be written as fluted. Namely, 	using first-order tautologies,  each conjunct in (7a) 
can be equivalently written as follows:
\begin{eqnarray*}
	\bigwedge_{i=0,2, j=0,1,2,3 }\forall x (Cx \wedge c_{ij}(x) \rightarrow  \forall y (b(x,y)\wedge c_{i+1,j}(y) \rightarrow \bigvee_{C': (C,C') \in {\cal C}_H}  C'y))\wedge\\
	\bigwedge_{i=1,3, j=0,1,2,3 }\forall x (Cx \wedge c_{ij}(x) \rightarrow  \forall y (r(x,y)\wedge c_{i+1,j}(y) \rightarrow \bigvee_{C': (C,C') \in {\cal C}_H}  C'y)),
\end{eqnarray*}
and similarly for (7b). 
%
%
Let $\eta_{\boldsymbol{\cal C}}$ be the conjunction of $\phi_{grid}$ with the properties (6) and (7) written in $\FLt$, as explained.  It is routine to show that we have simultaneously reduced the plane tiling problem (respectively, the torus tiling problem) to the (finite) satisfiability problem. If $\boldsymbol{\cal C}$ tiles any of the spaces  $\Z \times \Z$ or $\Z_t \times \Z_t$, for some $t$, we expand the grids to our intended models. In the opposite direction, when $\fG\models \eta_{\boldsymbol{\cal C}}$ and $\fG$ is infinite we obtain a tiling of $\Z\times \Z$; in case $\fG$ is finite we obtain a tiling of $\Z_t\times \Z_t$ with $t$ divisible by 4. As a result we conclude that the satisfiability and the finite satisfiability problems for $\FLt{}$ with two transitive relations are undecidable.

\section{Table defining conjuncts in group (3) of the proof of Theorem~\ref{th:three}}
\begin{table}[htb]
	\begin{center}
		\begin{minipage}{14cm}
			\begin{minipage}{0.45\textwidth}
				\begin{center}
					
					$	\begin{array}{c|c|c|c|c}
					colour \quad & diag(x) &  border(x)	&  t(x,y) & colour'\\
					\hline
					
					d_{00}	&	e		&   l		& 	b	& 	c_{01}\\
					d_{00}	&	e		&   \neg l	& 	g	& 	c_{50}\\
					d_{14}	&	e		&   		& 	b	& 	c_{31}\\
					d_{22}	&	e		&  			& 	r	& 	c_{12}\\
					
					\hline
					
					
					d_{00}	&	\neg e	& 	 		& 	r	& 	d_{01}\\
					d_{01}	&			&   		& 	r	& 	d_{21}\\
					d_{21}	&			&   		& 	r	& 	d_{22}\\
					d_{22}	&	\neg e	&   		& 	b	& 	d_{23}\\
					
					d_{23}	&			&   		& 	b	& 	d_{13}\\
					d_{13}	&			&   		& 	b	& 	d_{14}\\
					d_{14}	&	\neg e	&   		& 	g	& 	d_{15}\\
					
					d_{15}	&			&   		& 	g	& 	d_{05}\\
					d_{05}	&			&   		& 	g	& 	d_{00}\\
					\hline
					
					d_{11}	&	\		&   		& 	b	& 	d_{10}\\
					d_{10}	&			&   \neg f	& 	b	& 	d_{20}\\
					
					d_{20}	&			&  \neg f	& 	b	& 	d_{25}\\
					
					d_{25}	&			&   		& 	r	& 	d_{24}\\
					d_{24}	&			&   		& 	r	& 	d_{04}\\
					
					d_{04}	&			&   		& 	r	& 	d_{03}\\
					
					d_{03}	&			&   		& 	g	& 	d_{02}\\
					d_{02}	&			&   		& 	g	& 	d_{12}\\
					d_{12}	&			&   		& 	g	& 	d_{11}\\
					
					\hline
					d_{20}	&			&   f	& 	r	& 	d_{00}\\
					d_{10}	&			&    f	& 	r	& 	d_{20}\\

				\end{array}$
				\end{center}
				\end{minipage}
				\hspace{1cm}	
				\begin{minipage}{0.45\textwidth}
				\begin{center}

				$	\begin{array}{c|c|c|c|c}
				colour \quad & diag(x) &  border(x)	&  t(x,y) & colour'\\
				\hline
				&			&  			& 		& 	\\
				c_{01}	&	 e'		&   		& 	b	& 	d_{11}\\
				c_{20}	&	e'		&   		& 	g	& 	d_{03}\\
				c_{42}	&	e'		&   		& 	r	& 	d_{25}\\
				\hline
				
				c_{01}	&	\neg e'	&   		& 	b	& 	c_{11}\\
				c_{11}	&			&   		& 	b	& 	c_{10}\\
				c_{10}	&			&   		& 	b	& 	c_{20}\\
				
				c_{20}	&	\neg e'	&   		& 	g	& 	c_{30}\\
				c_{30}	&			&   		& 	g	& 	c_{32}\\
				c_{32}	&			&   		& 	g	& 	c_{42}\\
				
				c_{42}	&	\neg e'	&   		& 	r	& 	c_{52}\\
				c_{52}	&			&   		& 	r	& 	c_{51}\\
				c_{51}	&			&   		& 	r	& 	c_{01}\\
				\hline

				c_{12}	&			&   		& 	g	& 	c_{02}\\
				c_{02}	&			&   \neg l	& 	g	& 	c_{00}\\
				c_{00}	&			&   \neg l	& 	g	& 	c_{50}\\
				
				c_{50}	&			&   		& 	b	& 	c_{40}\\
				c_{40}	&			&   		& 	b	& 	c_{41}\\
				c_{41}	&			&   		& 	b	& 	c_{31}\\
				
				c_{31}	&			&   		& 	r	& 	c_{21}\\
				c_{21}	&			&   		& 	r	& 	c_{22}\\
				c_{22}	&			&   		& 	r	& 	c_{12}\\
				\hline
				
				c_{02}	&			& l 		& 	b	& 	c_{00}\\
				c_{00}	&			& l 		& 	b	& 	c_{01}				
			\end{array}$
			\end{center}
			\end{minipage}
			\end{minipage}
			\end{center}
			
			\caption{Various combinations of the literals in conjuncts of the form \eqref{rgb:witnesses}; empty entries in columns $diag(x)$ or $border(x)$ mean $\top$. }
			\label{table:witnesses}
		\end{table}

	\end{document}